\documentclass{article}
\usepackage{amsmath,amsfonts,amsthm,amssymb}
\usepackage{setspace}
\usepackage{fancyhdr}
\usepackage{lastpage}
\usepackage{extramarks}
\usepackage{chngpage}
\usepackage{soul,color}
\usepackage{graphicx,float,wrapfig}
\usepackage{tikz}
\usepackage{tikz-cd}
\usetikzlibrary{calc}
\usepackage{Base}
\usetikzlibrary{automata, positioning, arrows, shapes.geometric,math}
\newcommand*{\rows}{4}
\newcommand{\Ima}{\operatorname{im}}
\usepackage{cleveref}
\usepackage{authblk}
\usepackage{booktabs}
\usepackage[style=numeric,maxbibnames=99,maxalphanames=5,doi=false,url=false,isbn=false,sorting=none,eprint=false,backend=bibtex]{biblatex}
\hypersetup{colorlinks=true,linkcolor=red,citecolor=magenta,urlcolor=blue}
\usepackage[margin=1in]{geometry}

\let\ker\relax

\DeclareMathOperator{\ker}{ker}

\let\Set\relax
\DeclareMathOperator{\Set}{Set}

\newcommand{\ZZ}{\mathbb{Z}}

\newcommand{\code}{\textnormal{code}}

\newcommand{\euclid}{\textnormal{euclid}}
\newcommand{\intt}{\mathrm{int}}

\title{Transform Arbitrary Good Quantum LDPC Codes into Good Geometrically Local Codes in Any Dimension}
\date{}
\author[1]{Xingjian Li\thanks{lxj22@mails.tsinghua.edu.cn}}
\author[2,3]{Ting-Chun Lin\thanks{til022@ucsd.edu}}
\author[3]{Min-Hsiu Hsieh\thanks{min-hsiu.hsieh@foxconn.com}}
\affil[1]{Department of Computer Science and Technology, Tsinghua University, Beijing}
\affil[2]{Department of Physics, University of California San Diego, CA}
\affil[3]{Foxconn Research, Taipei, Taiwan}
\begin{document}
\begin{spacing}{1.1}
\maketitle \thispagestyle{empty}

\begin{abstract}

{Geometrically local quantum codes, comprised of qubits and checks embedded in $\R^D$ with local check operators, have been a subject of significant interest. A key challenge is identifying the optimal code construction that maximizes both dimension and distance. Recent advancements have produced several constructions, but these either depend on specific good quantum low-density parity-check (qLDPC) codes or are limited to three dimensions. In this work, we introduce a construction that can transform any good qLDPC code into an optimal geometrically local quantum code. Our approach hinges on a novel procedure that extracts a two-dimensional structure from an arbitrary three-term chain complex. We expect that this procedure will find broader applications in areas such as weight reduction and the geometric realization of chain complexes.}

\end{abstract}

\section{Introduction}
In recent years, quantum coding theory has become an area with significant progress in both theory and practice. Among various quantum codes, the quantum low-density parity-check (qLDPC) codes have drawn much attention. On the practical side, the qLDPC codes are favored since their stabilizers only act on a few qubits. The low-density property is experimental-friendly since quantum devices are sensitive to noise and measurement errors. On the theoretical side, there have been many constructions that achieved optimal asymptotically linear dimension and distance recently~\cite{panteleevAsymptoticallyGoodQuantum2022,leverrierQuantumTannerCodes2022,dinurGoodQuantumLDPC2023}. Interestingly, these constructions have deep connections with high dimensional expanders.

However, to connect the theoretical development with practice, there are still several barriers. One of the barriers is that the constructions of optimal codes are based on expanders, thus the checks do not have a geometrically local embedding in the Euclidean space.
For certain applications, it is preferable for the code to have a local embedding, allowing the system to avoid nonlocal operations.

If we take the geometry constraints into account, it is known that we cannot achieve linear dimension and distance simultaneously. In particular, Bravyi and Terhal~\cite{bravyiNogoTheoremTwodimensional2009} provided an upper bound on code distance, and Bravyi, Poulin, and Terhal~\cite{bravyiTradeoffsReliableQuantum2010} showed an upper bound on the distance and rate tradeoff for geometrically local codes. This was far from the known code constructions at the time and closing the gap has been a major open problem for coding theorists and physicists.

With the recent progress in qLDPC codes, the gap has been closed following two separate lines of work. In the first approach, Portnoy~\cite{portnoyLocalQuantumCodes2023} provided a construction that is optimal up to polylogs for any dimension $D\geq 3$ based on good qLDPC codes. This method involves transforming a code into a manifold, extracting a 2D complex, subdividing it, and then embedding the structure into $\R^D$. Following this idea, Lin, Wills, and Hsieh~\cite{linGeometricallyLocalQuantum2023} simplified the construction by identifying a specific family of good qLDPC codes that directly induce a 2D complex, bypassing the need for manifolds. While the original construction is only optimal up to polylogs, an improved embedding result on simplicial complexes~\cite{LP-future} gives optimal codes without polylogarithmic loss.
In a separate line of study, Williamson and Baspin~\cite{williamsonLayerCodes2023} introduced a layered construction that can turn any good qLDPC code into a geometrically local quantum code that is optimal in 3D.

We observe that while the constructions in \cite{portnoyLocalQuantumCodes2023,linGeometricallyLocalQuantum2023} work for arbitrary dimensions,
    the constructions only apply to certain initial good qLDPC codes.
On the other hand, while the constructions in \cite{williamsonLayerCodes2023} work for arbitrary good qLDPC codes,
    the construction so far only works in 3D.

In this work, we reconcile the two aspects and provide a construction that can map any good qLDPC code to a geometrically local quantum code with optimal parameters for any dimension $D\geq 3$.
The overall construction is similar to~\cite{portnoyLocalQuantumCodes2023,linGeometricallyLocalQuantum2023}
    but the key new ingredient is a procedure that extracts a 2D complex directly from an arbitrary quantum code.

\subsection{Main Contributions}

Our main result can be stated through the following theorem:
\begin{theorem}\label{thm:main}
    Given any family of quantum LDPC codes with linear dimension, linear distance, and linear energy barrier, there exists a family of $D$-dimensional geometrically local quantum code with code dimension $k=\Omega(n^{\frac{D-2}{D}})$, code distance $d=\Omega(n^{\frac{D-1}{D}})$, and energy barrier $\calE=\Omega(n^{\frac{D-2}{D}})$.
\end{theorem}

Our method surpasses previous constructions in its generality, accommodating arbitrary quantum LDPC codes in arbitrary dimensions. The crucial advancement lies in the technique presented in Section~\ref{sec:2DStruc} and Appendix~\ref{sec:app}, which establishes a novel connection between an arbitrary quantum code and a 2D complex. Note that the resulting code family also admits a distance-rate tradeoff where $d=\Omega((n/k)^{\frac{D-1}{2}})$ and $\calE=\Omega((n/k)^{\frac{D-2}{2}})$ for any $k\geq\Omega(n^{\frac{D-2}{D}})$ through copying, as demonstrated in~\cite{linGeometricallyLocalQuantum2023}.

The geometrization process says that for every chain complex (or quantum CSS code), there is an associated simplicial/square complex, such that for many geometric operations\footnote{   Morally, we hope to say that this applies to all geometric operation of the simplicial complex. The caveat that let us retreat from that strong position is that the geometric operation has to respect certain even-odd conditions as we will see in \Cref{sec:app}.   Nevertheless, these even-odd conditions should be thought of as mild restrictions, and do not pose an issue for most studies on asymptotic behaviors.} performed on the simplicial/square complex,   such as subdivision and sliding of simplices,    there are corresponding operations on the chain complex. The construction of the geometrically local quantum code is a concrete application where we find the corresponding square complex of a good quantum LDPC code, then subdivide the quantum code in the same way we subdivide a square complex.

In algebraic topology, it is well-known that a chain complex can be obtained from a manifold or a CW complex (cellular complex), where a CW complex is said to have bounded degree if each cell is attached to $O(1)$ number of other cells.\footnote{ We explain what the $O(1)$ is respecting to. Usually, the discussion of bounded degree contains a sequence of CW complexes with an increasing number of cells. So $O(1)$ is respecting to this family, meaning that the degree is bounded by some constant independent of the number of cells.}. Conversely, all chain complexes with chosen bases can be realized by a CW complex \cite{9829}. However, the CW complex constructed from a bounded degree chain complex using this method does not have a bounded degree.  For certain applications, such as constructing manifolds with systolic freedom \cite{freedmanBuildingManifoldsQuantum2021} or performing weight reduction (refer to the beginning of page 2 in \cite{hastingsQuantumWeightReduction2021} for the context), it is important for the CW complex to have bounded degree. Therefore, our new technique in geometrizing arbitrary chain complexes fills a crucial gap in existing approaches, offering new avenues for these problems.

\subsection{Technical Overview}

It is known that a good quantum LDPC code $Q$ cannot have a geometrically local embedding, since its underlying Tanner graph $\calT(Q)$ is a good expander, where the Tanner graph is the tripartite graph with $X$ checks, $Z$ checks, and qubits as vertices, connected if a check acts on a qubit. Our target is to obtain another code $Q_L$ with a Tanner graph $\calT(Q_L)$ such that the Tanner graph has an embedding in $\R^D$, where the embedding has local neighbors and bounded density in $\R^D$, and $L$ is the subdivision parameter we will specify later.

To overcome the aforementioned barrier of embedding expanders, we can subdivide an expander $G$ to obtain a graph with a geometrically local embedding.
To be precise, we obtain the graph $G_L=(V_L,E_L)$ from $G$ by replacing every edge in $G$ with a length $L$ chain for some large enough $L$. However, directly subdividing the Tanner graph
$\calT(Q)$ does not yield a valid Tanner graph for a quantum code
 $Q_L$, as it must preserve the commutation relation between $Z$ and $X$ stabilizers, the readers can find an example in~\Cref{fig:1dsubdiv}. The core problem is that the commutation relation contains 2D information, so merely considering a 1D structure is insufficient for our needs.
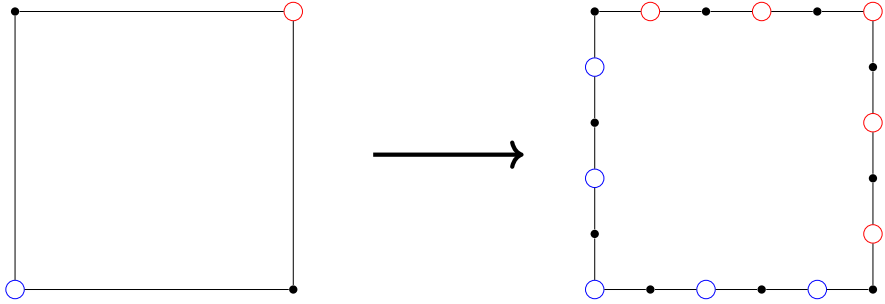
\begin{figure}
    \centering
    \begin{subfigure}[b]{0.25\textwidth}

        \centering
        \resizebox{\linewidth}{!}{
        \begin{tikzpicture}
        \node at (0,0) [circle,draw,blue](v00){};
        \node at (5,0) [circle,fill,inner sep=1.5pt](v50){};
        \node at (0,5) [circle,fill,inner sep=1.5pt](v05){};
        \node at (5,5) [circle,draw, red](v55){};
        \draw (v00)--(v05);
        \draw (v50)--(v00);
        \draw (v55)--(v05);
        \draw (v50)--(v55);
    \end{tikzpicture}
        }
    \end{subfigure}
    \qquad\tikz[baseline=-\baselineskip]\draw[ultra thick,->] (0,1.5) --  (2,1.5);\qquad
    \begin{subfigure}[b]{0.25\textwidth}

        \centering
        \resizebox{\linewidth}{!}{
        \begin{tikzpicture}
        \foreach \row in {0,1,2}
        {
            \foreach \column in {0,1,2}
            {
                \pgfmathtruncatemacro{\rw}{2*\row}
                \pgfmathtruncatemacro{\cl}{2*\column}
                \node at (2*\row,2*\column) [circle, draw,blue](v\rw\cl){};
                \pgfmathtruncatemacro{\cl}{2*\column+1}
                \node at (2*\row,2*\column+1) [circle,fill,inner sep=1.5pt](v\rw\cl) {};
                \pgfmathtruncatemacro{\rw}{2*\row+1}
                \node at (2*\row+1,2*\column+1) [circle, draw,red](v\rw\cl){};
                \pgfmathtruncatemacro{\cl}{2*\column}
                \node at (2*\row+1,2*\column) [circle,fill,inner sep=1.5pt](v\rw\cl) {};
            };
        };

    \foreach \column in {0,1,...,4}
        {
            \pgfmathtruncatemacro{\cl}{\column+1}
            \draw (v0\column)--(v0\cl);
            \draw (v5\column)--(v5\cl);
        }

    \foreach \row in {0,1,...,4}
        {
            \pgfmathtruncatemacro{\rw}{\row+1}
            \draw (v\row0)--(v\rw0);
            \draw (v\row5)--(v\rw5);
        }
    \filldraw[white] (0.5,0.5) rectangle (4.5,4.5);
    \end{tikzpicture}
        }
    \end{subfigure}
    \caption{A failed attempt if we only subdivide the tanner graph $\calT(Q)$ in 1D, taking one pair of $X$ and $Z$ check for example,  where blue, black, red vertices represent $X$ check, qubit, $Z$ check respectively. We can check that the $X$ and $Z$ check at the left-up corner do not commute since they only share one common qubit. Other assignments of checks and qubits will fail due to similar reasons.}
    \label{fig:1dsubdiv}
\end{figure}

Our construction resolves the problem by deriving a 2D structure $\calS(Q)$ from the Tanner graph $\calT(Q)$. Since the $X$ and $Z$ stabilizers of the code $Q$ commutes, in the Tanner graph $\calT(Q)$, each pair of $X$ and $Z$ stabilizer shares an even number of common neighbor qubits. By this observation, we can pair up the common neighbor qubits, and obtain a square face set $F$ where each face $f\in F$ consists of one $X$ check, one $Z$ check, and a pair of their common neighbor qubits in the Tanner graph. We can obtain a 2D square complex $\calS(Q)=(V,E,F)$ from the Tanner graph $\calT(Q)=(V,E)$ and include the additional face set $F$. For our final code construction, we will make some modifications to the face set $F$, and we will discuss the details in~\Cref{sec:2DStruc}.

 Now we subdivide the structure $\calS(Q)$ to obtain a Tanner graph $\calT(Q_L)$. Taking inspiration from the surface code construction~\cite{kitaevFaulttolerantQuantumComputation2003}, it is natural to subdivide the square face in $\calS(Q)$ to an $L\times L$ size grid and assign $X$ checks, $Z$ checks, qubits correspondingly as in surface codes, shown in~\Cref{fig:codediv}. It is known that after subdivision, the Tanner graph $\calT(Q_L)$ has a geometrically local embedding~\cite{portnoyLocalQuantumCodes2023,LP-future}.
 \begin{figure}
    \centering
    \begin{subfigure}[b]{0.25\textwidth}

        \centering
        \resizebox{\linewidth}{!}{
        \begin{tikzpicture}
        \node at (0,0) [circle,draw,blue](v00){};
        \node at (5,0) [circle,fill,inner sep=1.5pt](v50){};
        \node at (0,5) [circle,fill,inner sep=1.5pt](v05){};
        \node at (5,5) [circle,draw, red](v55){};
        \draw (v00)--(v05);
        \draw (v50)--(v00);
        \draw (v55)--(v05);
        \draw (v50)--(v55);
    \end{tikzpicture}
        }
    \end{subfigure}
    \qquad\tikz[baseline=-\baselineskip]\draw[ultra thick,->] (0,1.5) --  (2,1.5);\qquad
    \begin{subfigure}[b]{0.25\textwidth}

        \centering
        \resizebox{\linewidth}{!}{
        \begin{tikzpicture}
        \foreach \row in {0,1,2}
        {
            \foreach \column in {0,1,2}
            {
                \pgfmathtruncatemacro{\rw}{2*\row}
                \pgfmathtruncatemacro{\cl}{2*\column}
                \node at (2*\row,2*\column) [circle, draw,blue](v\rw\cl){};
                \pgfmathtruncatemacro{\cl}{2*\column+1}
                \node at (2*\row,2*\column+1) [circle,fill,inner sep=1.5pt](v\rw\cl) {};
                \pgfmathtruncatemacro{\rw}{2*\row+1}
                \node at (2*\row+1,2*\column+1) [circle, draw,red](v\rw\cl){};
                \pgfmathtruncatemacro{\cl}{2*\column}
                \node at (2*\row+1,2*\column) [circle,fill,inner sep=1.5pt](v\rw\cl) {};
            };
        };
    \foreach \row in {0,1,...,5}
    {
        \foreach \column in {0,1,...,4}
        {
            \pgfmathtruncatemacro{\cl}{\column+1}
            \draw (v\row\column)--(v\row\cl);
        }
    }

    \foreach \column in {0,1,...,5}
    {
        \foreach \row in {0,1,...,4}
        {
            \pgfmathtruncatemacro{\rw}{\row+1}
            \draw (v\row\column)--(v\rw\column);
        }
    }
    \end{tikzpicture}
        }
    \end{subfigure}
    \caption{Constructing new code by subdivision, where blue, black, red vertices represent $X$ checks, qubits, $Z$ checks respectively.}
    \label{fig:codediv}
\end{figure}
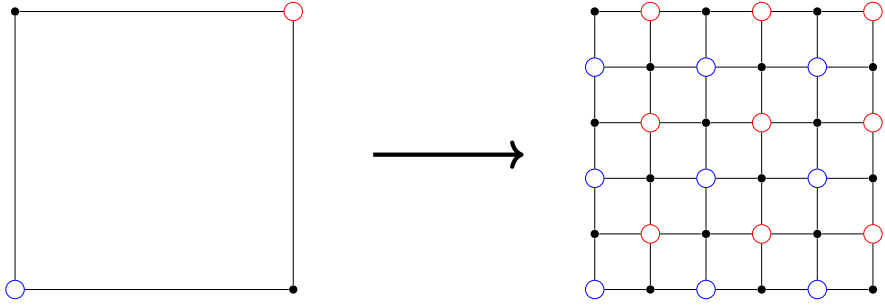

To analyze the properties of the new code $Q_L$, we relate it with the original code $Q$. Note that every CSS code has a corresponding chain complex. We use $X$ to denote the chain complex of the qLPDC code $Q$ and $X_L$ for our target subdivision code $Q_L$. For the ease of our analysis, we introduce a chain map $\calF$ from $X\colon\F_2^{X(0)}\to\F_2^{X(1)}\to\F_2^{X(2)} $ to $X_L\colon\F_2^{X_L(0)}\to\F_2^{X_L(1)}\to\F_2^{X_L(2)} $ shown in the following commutation diagram.
\[
    \begin{tikzcd}
    \F_2^{X(0)}\arrow[r,"\delta_0"]\arrow[d,"\calF_0"] &\F_2^{X(1)}\arrow[r,"\delta_1"] \arrow[d,"\calF_1"]&\F_2^{X(2)}\arrow[d,"\calF_2"]\\
    \F_2^{X_L(0)}\arrow[r,"\delta_0"] &\F_2^{X_L(1)}\arrow[r,"\delta_1"] &\F_2^{X_L(2)}
    \end{tikzcd}
\]
We will show how to construct such a chain map in~\Cref{sec:subdiv}.

In~\Cref{subsec:prop}, we will analyze the properties of our code. Our analysis has a similar flavor as the proof in~\cite{linGeometricallyLocalQuantum2023}, as we both used the expansion properties of the chain complex $X_L$ to obtain lower bounds for properties of code $Q_L$. The difference is that the original proof utilized the local product structure of the balanced product code, while our proof applies more generally to arbitrary local structure. 
The detailed analysis of the expansion of local structure can be found in~\Cref{sec:exp}.

\subsection{Related Work}

We briefly compare our construction to the recent constructions~\cite{portnoyLocalQuantumCodes2023,linGeometricallyLocalQuantum2023,williamsonLayerCodes2023}.
A detailed historical review of the work on geometrically local codes can be found in~\cite{linGeometricallyLocalQuantum2023}.

The first construction of (almost) good geometrically local codes is by Portnoy~\cite{portnoyLocalQuantumCodes2023}. It would first geometrize the code by mapping the code into a manifold based on the work of Freedman and Hastings~\cite{freedmanBuildingManifoldsQuantum2021}, take the nerve of the manifold to obtain a 2D simplicial complex, subdivide it, and embed the 2D simplicial complex into $\R^D$ using the work of Gromov and Guth~\cite{gromov2012generalizations}. Due to the use of~\cite{freedmanBuildingManifoldsQuantum2021}, this method only works for codes that have a sparse $\Z$-lift.

The following work of~\cite{linGeometricallyLocalQuantum2023} simplifies Portnoy's construction by directly applying subdivision on the code, avoiding the step of turning the code into a manifold. However, this method only works for balanced product codes~\cite{breuckmannBalancedProductQuantum2021} (also known as lifted product codes~\cite{panteleevQuantumLDPCCodes2022}). Our result generalizes~\cite{linGeometricallyLocalQuantum2023} from the specific balanced product code to any code.

The work on layer codes~\cite{williamsonLayerCodes2023} provides a different approach that applies to all initial quantum codes and has a straightforward local embedding in 3D. However, it currently only works in 3D, whereas our result can be applied to any embedding dimension.

We summarize the previous result and our result in the following table:
\begin{table}[H]
  \centering
  \caption{Comparison of the constructions}
  \label{table:comparison}
  \begin{tabular}{ccccc}
    \toprule
    & \textbf{BPT bound} & \textbf{Energy barrier} & \textbf{Arbitrary code} & \textbf{All dimensions} \\
    \midrule
    Portnoy & \checkmark $^\sharp$ & \checkmark $^\natural$ & Almost$^\flat$ & \checkmark \\
    Lin-Wills-Hsieh & \checkmark $^\sharp$ & \checkmark & Only balanced product & \checkmark \\
    Williamson-Baspin & \checkmark & \checkmark & \checkmark & Currently in 3D \\
    Our result & \checkmark &\checkmark &\checkmark &\checkmark\\
    \bottomrule
  \end{tabular}
  \flushleft

  $^\sharp$After using the optimal embedding \cite{LP-future}.
  $^\natural$Not stated, but it holds.
  $^\flat$Requires a sparse $\Z$-lift \cite[Definition 1.2.2]{freedmanBuildingManifoldsQuantum2021}.
\end{table}

\subsection{Discussions and Further Directions}

We also expect our construction to find more applications in both algebraic topology and coding theory. We point out several possible directions for future work.

\paragraph*{Geometrize and Embed Longer Chain Complexes:}

This work focuses on quantum CSS codes,
    which are 3-term chain complexes over $\F_2$,
    and endow them with 2D geometrical structures.
We expect a similar result holds in higher dimensions over arbitrary finite fields or $\Z$.
In particular, we can endow a $t+1$-term chain complex
    with a $t$-dimensional CW complex structure after pairing.
Furthermore, when the chain complex has bounded degree with bounded entry values,
    the CW complex also has bounded degree.
When combined with the embedding theorem for simplicial complexes
    this provides a way to embed a $t+1$-term chain complex in $\R^D$ with $D > t$,
    which likely saturates the generalization of the BPT bound.
Some further thoughts are expressed in \Cref{sec:app}.

\paragraph*{Weight Reduction:}

Given a qLDPC code, can one convert it into another qLDPC code with similar properties while having a smaller check weight?
This problem was first answered by Hastings' in \cite{hastingsQuantumWeightReduction2021}.
The intuition behind the work is elegantly articulated in the introduction of the paper. The reasoning is that there seem to be a general procedure that transforms quantum codes into manifolds. (Note that the reverse process is straightforward, as manifolds can always be converted into quantum codes.) Given that manifolds can naturally be refined such that each cell is only attached to a constant number of other cells, this suggests a method of weight reduction.

However, a challenge in implementing this idea is that the procedure for transforming a code into a manifold, as described in \cite{freedmanBuildingManifoldsQuantum2021}, does not work for all codes. Nevertheless, Hastings found an alternative approach to achieve weight reduction, as detailed in \cite{hastingsQuantumWeightReduction2021}.

We believe that given our procedure to endow an arbitrary quantum code with a 2D geometrical structure, we can fully realize Hastings' intuition.  We will describe the details in future work.

\paragraph*{Simpler Manifold Construction:}

As discussed in weight reduction, it is believed that there exists a general procedure for transforming a code into a manifold.
This procedure is valuable for achieving weight reduction and for constructing manifolds with systolic freedom, which was the main motivation in \cite{freedmanBuildingManifoldsQuantum2021}.

In \cite{freedmanBuildingManifoldsQuantum2021},
    the manifold is constructed by gluing the handlebodies.
To make sure the handlebody glues,
    one computes the obstructions and makes sure they vanish.
These are standard techniques in topology
    but computing the obstructions requires some work.

The 2D geometric structure described in this paper
    can facilitate showing the existence of the glue.
We can obtain a glue
    using the immersion of the 2D structure in $\R^5$ without self-intersections,
    which follows from a general-position argument.
This immersion leads to a natural framing
    and induces a $4$-dimensional CW complex (which can be turned into a simplicial) whose
    2-cells $\to$ 3-cells $\to$ 4-cells is isomorphic to the chain complex over $\F_2$
    (up to some extra 2-cells with a zero map to the 3-cells).
This complex can then be transformed into a $8$-dimensional manifold whose
    2-cells $\to$ 3-cells $\to$ 4-cells is isomorphic to the chain complex over $\F_2$
    (up to some extra 2-cells and 4-cells with zero maps to the 3-cells).
Further details will be described elsewhere.

\section{Preliminaries}
% The following simple lemma might be of help

In this section, we will first introduce some basic facts about chain complexes and quantum CSS codes, and their relationship. We will also formally define a square complex, and how to embed it into the Euclidean space via subdivision.

\subsection{Chain Complexes}\label{sec:chain-complex}
Chain complexes offer an intuitive structure for studying quantum CSS codes. Within this framework, we can express the properties of the CSS code using the language of chain complexes, covering aspects such as dimension, distance, and energy barrier of a given code. We mainly consider chain complexes over the finite field $\F_2$.

\begin{definition}[Chain complex]
    A chain complex $X$ consists of a sequence of vector spaces $\F_2^{X(i)}$ generated by sets $X(i)$, along with linear maps $\delta_i\colon \F_2^{X(i)}\to \F_2^{X(i+1)}$ known as coboundary operators, where the coboundary operators satisfy
    \begin{align*}
        \delta_{i+1}\delta_i=0.
    \end{align*}
\end{definition}

By considering dual maps, one can also define the dual chain complex consisting of boundary operators. In our context, there is a canonical basis of $\F_2^{X(i)}$ labeled by the elements in $X(i)$. Under this basis, the boundary operator $\partial_i\colon \F_2^{X(i)}\to\F_2^{X(i-1)}$ can be written as $\partial_i=\delta_{i-1}^T$. The boundary operators will satisfy:
\begin{align*}
    \partial_{i-1}\partial_i=0.
\end{align*}

We introduce some standard definitions. Elements in the kernel of the (co)boundary operators are called (co)cycles:
\begin{align*}
    Z_i \coloneq \ker \partial_i=\{c_i\in\F_2^{X(i)}:\partial_ic_i=0\},\quad
    Z^i\coloneq\ker \delta_i=\{c_i\in\F_2^{X(i)}:\delta_ic_i=0\}.
\end{align*}
Elements in the image of the (co)boundary operators are called (co)boundaries:
\begin{align*}
    B_i\coloneq \Ima \partial_{i+1}=\{\partial_{i+1}c_{i+1}:c_{i+1}\in\F_2^{X(i+1)}\},\quad   B^i\coloneq \Ima \delta_{i-1}=\{\delta_{i-1}c_{i-1}:c_{i-1}\in\F_2^{X(i-1)}\}.
\end{align*}

A chain is called exact if $Z_i=B_i$ for all $i$. We can also define an exact cochain similarly.
%We can also define a cochain is exact similarly.

\subsection{Quantum CSS Code}
A quantum CSS code $Q$ is defined by two classical codes $C_x, C_z$ represented by their parity check matrices $H_x\colon \F^n_2\to \F^{m_z}_2$ and $H_z\colon \F^n_2\to\F^{m_x}_2$ that satisfies $H_x H_z^T =0$. Here $n,m_x,m_z$ corresponds to the number of qubits, $X$ checks, and $Z$ checks respectively. It is well known that the CSS code naturally corresponds to a chain complex as follows:
\[
\begin{tikzcd}
    \F_2^{m_x}\arrow[r,"\delta_0=H_z^T"]&\F_2^n\arrow[r,"\delta_1=H_x"]&\F_2^{m_z}.
    % \F_2\arrow[r,"\delta_{-1}"]&\F_2^V\arrow[r,"\delta_0"]&\F_2^E.
\end{tikzcd}
\]

The $X$ and $Z$ logical operators correspond to the code $C_x$ and $C_z$, and $X$ and $Z$ stabilizers correspond to the code $C_z^{\perp}$ and $C_x^{\perp}$. The code dimension is defined by $k=\dim C_x-\dim C_z^{\perp}=\dim C_z-\dim C_x^{\perp}$. The code distance $d=\min(d_x,d_z)$ where
\begin{align*}
    d_x=\min_{c_x\in C_x-C_z^\perp}|c_x|,\quad d_z=\min_{c_z\in C_z-C_x^\perp}|c_z|.
\end{align*}

We would also consider the energy gap $\calE$ of the system. The energy gap is the minimum energy required to change an all-zero codeword to a nontrivial codeword by flipping one bit at a time. In our context, this energy is related to the number of violated checks. For any vector $c\in \F_2^n$, we define its $X$ energy as $\epsilon_x(c_x)=|H_xc_x|$. A sequence of vectors $\gamma_{a\to b}=(c_0=a,c_1,\dots,c_t=b)$ constitutes a walk from $a$ to $b$ if for each $i\in[0,t]$, $|c_i-c_{i+1}|=1$. The $X$ energy of a walk is defined by $\epsilon_x(\gamma_{a\to b})=\max_{c\in\gamma_{a\to b}}\epsilon_x(c)$, and the $X$ energy gap of $\calE_x$ is defined by
\begin{align*}
    \calE_x=\min_{\gamma_{0\to c},c\in C_x-C_z^\perp}\epsilon_x(\gamma_{0\to c}).
\end{align*}
Similarly, we can define its $Z$ energy gap $\calE_z$. The code's energy gap is defined by $\calE=\min(\calE_x,\calE_z)$.

We say a quantum code is a low-density parity-check (LDPC) code if each check acts with a constant number of qubits, and each qubit is acted by a constant number of checks. We call a quantum LDPC code good if it has asymptotic linear dimension and distance.

Another common way to describe a quantum CSS code $Q$ is through its corresponding Tanner graph $\calT(Q)=(V=V_0\cup V_1\cup V_2, E=E_0\cup E_1)$, which is also a 1-simplicial complex. We would map each $X$ check to a vertex in $V_0$, every qubit to a vertex in $V_1$, and every $Z$ check to a vertex in $V_2$. $E_0$ consists of  edges between vertices in $V_0$ and $V_1$, and there is an edge between $v_0\in V_0$ and $v_1\in V_1$ iff $H_z(v_0,v_1)=1$ in the parity check matrix of $C_z$. We can define $E_1$ similarly for the parity check matrix $H_x$. From the Tanner graph, we would also use level 0, level 1, and level 2 vertices to refer to the $X$ checks, qubits, and $Z$ checks respectively.

In this paper, we also consider codes that have additional geometric structures.  Specifically, these codes should have an embedding in the lattice $\mathbb{Z}^D$: each qubit and each check correspond to a specific location in $\mathbb{Z}^D$. We characterize the embedding map as follows:
\begin{align*}
  I_{\code\to\euclid}: \Set(Q) \to \ZZ^D,
\end{align*}
where the $\Set(Q) = [m_x] \sqcup [n] \sqcup [m_z]$ is the set of the qubits and the checks' labels. Here $[n]$ is a set of size $n$ which labels the canonical basis vectors of $\F_2^n$
  and $\sqcup$ is the disjoint union.

We call an embedding `$a$-geometrically-local' if the Euclidean distance between each check and the qubit it interacts with is at most $a$ in the embedding. Formally speaking, if $H_{ij}$ is nonzero, then $|I_{\text{code}\to\text{euclid}}(i)-I_{\text{code}\to\text{euclid}}(j)| \leq a$, where $|\cdot|$ represents the Euclidean distance.
Additionally, the embedding is said to have density $b$ if the number of qubits and checks located at each lattice point in $\mathbb{Z}^D$ is at most $b$. Our goal is to obtain an embedding with constant parameters $a = \Theta(1)$ and $b = \Theta(1)$.

\subsection{Square complex and its subdivision}
In this section, we provide the formal definition of a square complex and its $L$ subdivision.
\begin{definition}[Square complex in 2D]
    A two-dimensional square complex $\calS=(V,E,F)$ consists of a vertice set $V$, an edge set $E$, and a face set $F$, which satisfy the following conditions:
    \begin{itemize}
        \item For every element $e\in E$, $e=\{v_0,v_1\}$, where $v_0,v_1\in V, v_0\neq v_1$.
        \item For every element $f\in F$, $f=\{e_0,e_1,e_2,e_3\}$, where the four edges in $E$ form a square, with four vertices in $V$.
    \end{itemize}
\end{definition}

We can obtain the $L$ subdivided complex $\calS_L=(V_L,E_L,F_L)$ from $\calS$ by dividing every square face $f\in F$ to an $L\times L$ grid, as shown in~\Cref{fig:complexdiv}. The definition of $V_L, E_L, F_L$ is direct from the figure.

\begin{figure}
    \centering
    \includegraphics{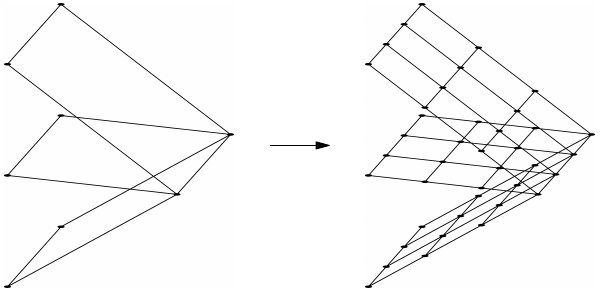}
    \caption{$L$ subdivision of square complex $\calS$ on each face $f$.(Figure 5 in~\cite{linGeometricallyLocalQuantum2023})}
    \label{fig:complexdiv}
\end{figure}
\subsection{Embedding via Subdivision}

To embed our subdivided code into the Euclidean space $\Z^D$, we would use the following embedding result from~\cite{portnoyLocalQuantumCodes2023}.
\begin{theorem}\label{thm:embedding}
For any $L$-subdivided square complex $\calS_L=(V_L, E_L, F_L)$ from a square complex $\calS=(V,E,F)$, there exists an embedding map $I_{\mathrm{square}_L\to \mathrm{euclid}}\colon V_L\to \Z^D$ with constant $a,b=\Theta(1)$ such that for $L=\Theta(|V|^{\frac{1}{D-2}}\polylog(|V|))$:
\begin{enumerate}
    \item Geometrically local: For all adjacent vertices on the complex $\{v_0,v_1\}\in E_L$, the distance between corresponding points in $\Z^D$ is bounded, i.e. $|I_{\mathrm{square}_L\to \mathrm{euclid}}(v_0)-I_{\mathrm{square}_L\to \mathrm{euclid}}(v_1)|\leq a$.
    \item Bounded density: The number of vertices at each point in $\Z^D$ is bounded,\\ i.e. $\forall x\in \Z^D, |I_{\mathrm{square}_L\to \mathrm{euclid}}^{-1}(x)|\leq b$.
\end{enumerate}
\end{theorem}

We are aware of an upcoming result by Lin and Portnoy~\cite{LP-future} that removes the polylog factor in the embedding. Thus, we will not include the polylog factor in our theorem statement below.

\section{2D Structure from Arbitrary QLDPC Code}\label{sec:2DStruc}

In this section, we will derive a 2D structure from a quantum LDPC code $Q$. This structure is crucial for our construction of geometrically local quantum codes.

To obtain the 2D structure, we utilize the commutation relation of $X$ and $Z$ checks, $H_x H_z^T=0$.
The commutation relation implies that
    for every $Z$ check $v_0$ and $X$ check $v_2$ in $\calT(Q)$,
    they share an even number of common neighbors $N(v_0,v_2)\subset V_1$,
    which allows us to pair up their common neighbors.
For each pair $\{v_1, v_1'\}$,
    we form a square with vertices $v_0, v_1, v_1', v_2$.
We denote the resulting face set obtained by considering all pairs of $v_0, v_2$ as $F$, and we define our resulting square complex as $\calS(Q)=(V,E,F)$. We will refer to these faces as full faces, and they are colored in green in~\Cref{fig:subdiv}.

For certain quantum codes with parity checks $H_x$ and $H_z$, we can stop here with $\calS(Q)$ as the desired 2D structure.
However, there are some pathological (or unreasonable) codes that require further treatments
    as we will describe them later.
We say a quantum code is reasonable~\cite{hastingsQuantumWeightReduction2021},
    if it doesn't have a $Z/X$ codeword that is a subset of a $Z/X$ stabilizer.
This means as long as the distance of the code is larger than the weight,
    the code is reasonable.
In particular, all good qLDPC codes are reasonable.

We say the stabilizer generators in $H_x$ and $H_z$ are `minimal'
    if those stabilizers do not contain smaller stabilizers.
Formally, a X-stabilizer $c_x \in C_z^\perp$ is minimal,
    if for all $c'_x \in C_z^\perp$ and $\operatorname{supp} c'_x \subseteq \operatorname{supp} c_x$,
    we have $c'_x = 0$ or $c_x$.
Notice that if a stabilizer generator $c_x \in C_z^\perp$ from $H_z$ is not minimal,
    by definition, there exist $c'_x \in C_z^\perp - 0$ and $\operatorname{supp} c'_x \subsetneq \operatorname{supp} c_x$.
That means we can replace the generator $c_x$ with $c'_x$ and $c_x - c'_x$,
    while maintaining the same code $C_x$ and $C_z$.
Thus, given $H_x$ and $H_z$, we can repeat this process until all generators are minimal.
(This process terminates in finite rounds, because the stabilizer weight strictly decreases when we split the generators.)
In particular, if the initial parity check matrices $H_x$ and $H_z$ are LDPC code,
    this process induces parity check matrices that are LDPC and are minimal.

The nice property of reasonable codes with minimal $H_x$ and $H_z$ is that
    the \emph{link} of a $Z/X$ check $v\in V_0\cup V_2$
    is always connected.
This property is used later to show the subdivided code has good parameters.
Let $N^F(v)$ be the faces that contain $v$, $N^F(v)=\{f\in F\mid v\in f\}$,
    and let $N^E(v)$ be the edges that contain $v$, $N^E(v)=\{e\in E\mid v\in e\}$.
The link of $v$ is a graph with `vertices' $e \in N^E(v)$
    and `edges' induced by $f \in N^F(v)$ which connects the two `vertices' $v \in e \in f$, $v \in e' \in f$.
Note that for a $Z$ check $v_0$,
    every qubit $v_1$ neighbor to $v_0$
        can be identified to the `vertex' $\{v_0, v_1\} \in N^E(v_0)$.
Furthermore, every check $v_2$ that share a qubit with $v_0$
    can be identified to a set of `edges' of the form $\{\{v_0, v_1\}, \{v_0, v_1'\}\} \in N^F(v_0)$ where $v_1, v_1'$ are the paired qubits in the common neighbors $N(v_0, v_2)$.
In particular, the endpoints of the `edges' corresponds to the shared qubits $N(v_0, v_2)$.
        
The following claim is highly related to~\cite[Lemma 6]{hastingsQuantumWeightReduction2021}.

\begin{claim}
    For a reasonable quantum code with minimal $H_x$ and $H_z$, the link of each $Z/X$ check $v\in V_0\cup V_2$ is connected.
\end{claim}
\begin{proof}
    Given a $Z$ check that corresponds to the vertex $v_0$,
    let $(V'_{link}, E'_{link}) \subseteq (V_{link}, E_{link}) = (N^E(v_0), N^F(v_0))$ be a connected component.
    We want to show that $(V'_{link}, E'_{link}) = (V_{link}, E_{link})$.
    Consider the $Z$ operator that act on the qubits that corresponds to $V'_{link}$, $\{v_1: \{v_0, v_1\} \in V'_{link}\}$.
    We claim that it is a logical operator.
    To show that, it suffices to show that it commutes with all $X$ checks $v_2$ that share a qubit with $v_0$.
    As discussed above, $v_2$ corresponds to a set of `edges' in $E''_{link} \subseteq E_{link}$, which acts on qubits that correspond to their endpoints.
    Since $(V'_{link}, E'_{link})$ is a connected component,
        the number of endpoints of $E''_{link}$ in $V'_{link}$ is even.
    Thus, the $Z$ operator is a logical operator.

    Because the code is reasonable, and the $Z$ operator is a subset of the $Z$ stabilizer $v_0$,
        the $Z$ operator is a stabilizer.
    Because $H_z$ is minimal, that means the operator is exactly $v_0$.
    That means $(V'_{link}, E'_{link}) = (V_{link}, E_{link})$ as desired.
\end{proof}

For unreasonable codes or codes with non minimal $H_x$ or $H_z$, however, the link of a check is not always connected.
Therefore, we will include dummy faces so that the link of every check becomes connected.
This is important for the subdivided code to have the desired distance property.

The corresponding 2D structure of an unreasonable code will be a relaxed notion of a square complex,
    which we call the square subspace complex.
A square subspace complex $\Tilde{\calS}=(\Tilde{V},\Tilde{E},\Tilde{F})$ also consists of
    vertices $\Tilde{V}$, edges $\Tilde{E}$, and faces $\Tilde{F}$,
    but it is no longer required to be downward closed,
    i.e. a face $f\in\Tilde{F}$ may contain an edge that is not in $\Tilde{E}$.
Generally, one can complete the square subspace complex to form a square complex,
    by including edges and vertices to make it downward close.
Thus, one can view a square subspace complex $\Tilde{\calS}$
    as a part of a square complex $\calS_u=(V_u,E_u,F_u)$,
    where $\Tilde{V}\subseteq V_u$, $\Tilde{E}\subseteq E_u$, and $\Tilde{F}\subseteq F_u$.

We now add dummy faces to $\calS(Q)$ to obtain the square subspace complex $\Tilde{\calS}(Q)$ which is the desired 2D structure.
The vertices and the edges are the same as before, $\Tilde{V} = V$ and $\Tilde{E} = E$,
    while there are additional faces $\Tilde{F} \supseteq F$.
The goal is to make the link of every check $v \in V_0 \cup V_2$ connected.
To do so, we simply add a dummy face $f$ to every pair of edges $e, e' \in N^E(v)$.
The dummy face consists of two existing edges $e, e' \in \Tilde{E}$
    and two imaginary edges $e_* = \{v_*,v_1\}, e_*' = \{v_*,v_1'\} \notin \Tilde{E}$, where $v_1,v_1' \in V_1$,
    and $v_*\notin \Tilde{V}$ is an imaginary vertex that is introduced to define $f$.
An example of the dummy face $f$ is shown in the brown part in~\Cref{fig:subdiv}.
It is clear that by adding faces for every pair of edges $e, e' \in N^E(v)$,
    the link of $v$ is connected.

As discussed, any good quantum LDPC code is reasonable and does not require the addition of dummy faces.
However, for simplicity, we will include dummy faces regardless of whether the code is reasonable to avoid considering two separate cases.
The proof in our paper relies solely on the property that the link is connected
    and can be easily adapted for good quantum LDPC codes without dummy faces.

\begin{figure}
    \centering
        \begin{tikzpicture}[square/.style={regular polygon,regular polygon sides=4}]
            \fill[green,nearly transparent] (-3,0)--(0,0.5)--(3,1)--(0,1.5)--cycle;
            \fill[green,nearly transparent] (-3,0)--(0,-0.5)--(3,-1)--(0,-1.5)--cycle;
            \fill[brown,nearly transparent] (-3,0)--(0,-0.5)--(3,0)--(0,0.5)--cycle;
            \node at (-3,2) [square,draw,blue](x1){};
            \node at (-3,0) [square,draw,blue](x2){};
            \node at (-3,-2) [square,draw,blue](x3){};
            \node at (0,2.5) [circle, draw] (q1) {};
            \node at (0,1.5) [circle, draw] (q2) {};
            \node at (0,0.5) [circle, draw] (q3) {};
            \node at (0,-0.5) [circle, draw] (q4) {};
            \node at (0,-1.5) [circle, draw] (q5) {};
            \node at (0,-2.5) [circle, draw] (q6) {};
            \node at (3,1) [square,draw,red] (z1) {};
            \node at (3,-1) [square,draw,red] (z2) {};
            \node at (3,0) [square,dashed,draw,red](z3){};
            \draw (x1)--(q1);
            \draw (x1)--(q2);
            \draw (x1)--(q4);
            \draw (x1)--(q6);
            \draw (x2)--(q2);
            \draw (x2)--(q3);
            \draw (x2)--(q4);
            \draw (x2)--(q5);
            \draw (x3)--(q3);
            \draw (x3)--(q5);
            \draw (x3)--(q6);
            \draw (x3)--(q1);
            \draw (z1)--(q1);
            \draw (z1)--(q2);
            \draw (z1)--(q3);
            \draw (z2)--(q4);
            \draw (z2)--(q5);
            \draw (z2)--(q6);
            \draw [dashed] (z3)--(q3);
            \draw [dashed] (z3)--(q4);
            \node at (0,-0.9) {$v_1'$};
            \node at (0,0) {$f$};
            \node at (0,0.9) {$v_1$};
            \node at (3.5,0) {$v_{*}$};
            \node at (1.5,0.5) {$e_*$};
            \node at (1.5,-0.5) {$e_*'$};
            \node at (-0.6,0.55) {$e$};
            \node at (-0.6,-0.55) {$e'$};
            \node at (-3,3.2) {\large{ $V_0$}};
            \node at (0,3.2) { \large{ $V_1$}};
            \node at (3,3.2) {\large{ $V_2$}};
        \end{tikzpicture}
        \caption{An example Tanner graph, where the blue, black, red vertices correspond to the $X$ checks, qubits, $Z$ checks respectively, we label them as $V_0,V_1,V_2$. We also include some of the face neighbors of one of the blue check. The green squares are the faces from $S(Q)$ which are downward closed, while the brown square is a dummy face $f$ which contains imaginary edges $e_*, e_*'$ and vertices $v_*$. }
        \label{fig:subdiv}

\end{figure}
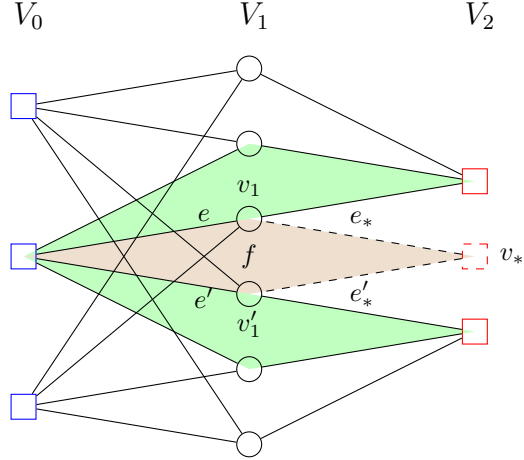

Through our construction in this section, we obtain a map $I_{\code\to \text{square}}$ that maps a qLDPC code (its corresponding chain complex $X$) to a square subspace complex $\Tilde{\calS}(X)$. By our construction, $\Tilde{\calS}(X)$ has bounded degree since every vertex is connected to a constant number of faces. We summarize the result in the following theorem:
\begin{theorem}
    Given an arbitrary bounded degree 3-term chain complex $X$ over $\F_2$,
    there exists an embedding to a bounded degree square subspace complex $I_{\code\to\mathrm{square}}: X\to \Tilde{\calS}(X)$. The square subspace complex has $\Tilde{\calS}(X)$ the same vertices and edges as the Tanner graph of $X$.
    It has a face structure where the link of its level 0 and level 2 vertices is connected.
\end{theorem}

Comparing with the previous result in~\cite{linGeometricallyLocalQuantum2023}, their construction can be seen as a special case of our construction, since their basis code from~\cite{panteleev2022asymptotically,dinurGoodQuantumLDPC2023} has a natural pairing for the qubits between the $X$ and $Z$ checks.
Our construction generalized their idea to an arbitrary 3-term chain complex and also provided a square subspace complex from the chain complex that might be of independent interest.

\section{Construct Geometrically Local Code from the 2D Structure}
\label{sec:subdiv}

In this section, we will utilize the subspace complex $\Tilde{\calS}(Q)$ from~\Cref{sec:2DStruc} to derive our geometrically local code construction. In summary, we would subdivide the faces $\Tilde{F}$ to obtain a new subspace complex $\Tilde{\calS}_L(Q)$. We will assign $Z$ checks, $X$ checks, and qubits to the vertices in $\Tilde{\calS}_L(Q)$ to obtain a Tanner graph $\calT(Q_L)$ of our subdivided code $Q_L$ in~\Cref{subsec:divqldpc}. We will show how to obtain a geometrically local embedding of $\calT(Q_L)$ through~\Cref{thm:embedding} in~\Cref{subsec:embed}, proving $Q_L$ is geometrically local.

Recall that each quantum CSS code $Q$ has a corresponding chain complex $X$. In the rest of our paper, we will specify our original code $Q$ as $Q(X)$, and the subdivided code $Q_L$ as $Q(X_L)$. We will also use level 0, level 1, and level 2 vertices to represent $X$ checks, qubits, and $Z$ checks respectively. We would also only consider the property of coboundary operators, as the boundary operators have the same proof by symmetry.

\subsection{Subdividing Arbitrary QLDPC Code}
\label{subsec:divqldpc}
In the 2D subspace complex $\Tilde{\calS}(Q)$, we have defined two types of faces as shown in~\Cref{fig:subdiv}: The full faces are colored green, and the dummy faces are colored brown. Taking an odd number $L$ as the subdivision parameter, we would first divide the full faces by mapping each square to a $L\times L$ grid as shown in~\Cref{subfig:fullface}. We label the vertices from $(0,0)$ to $(L,L)$ from down left to up right in the figure. Note that after dividing the full faces, for every dummy face, its two edges $e\in \Tilde{E}$ is subdivided to a $L$ chain. Depending on it being a dummy face attached to a level 0 or level 2 vertex, we would subdivide it to a $(L-1)\times(L-1)$ grid located on the square $(0,0)$ to $(L-1,L-1)$ or $(1,1)$ to $(L,L)$ respectively. Readers can refer to~\Cref{fig:squarediv} for an example. We can obtain a new subspace complex $\Tilde{\calS}_L(Q)$ through the process, where its vertex and edge set are clear from the context.\footnote{Actually, $\Tilde{\calS}_L(Q)$ is a square complex. We would not utilize its face structure, so we didn't define it.}

From the subdivided subspace complex $\Tilde{\calS}_L(Q)$, we can obtain a new chain complex $X_L$. For each subdivided face, we define the level 0, level 1, and level 2 vertices of the new chain complex $X_L$ as follows, under the convention that coordinates increase as going right and up in~\Cref{fig:squarediv}:
\begin{itemize}
    \item $X_L(0)$ are the set of vertices with $i,j$ both being even.
     \item $X_L(1)$ are the set of vertices with $i,j$ one being even, one being odd.
     \item $X_L(2)$ are the set of vertices with $i,j$ both being odd.
\end{itemize}
The corresponding vertices are colored blue, black, and red in~\Cref{fig:squarediv}.

The construction above provides us a Tanner graph $\calT(Q_L)$ with corresponding code $Q(X_L)$, with the coboundary operators $\delta_i$ defined by the adjacency matrices on the graph. It is easy to verify that $\delta_1\delta_0=0$, meaning that the $X$ and $Z$ stabilizers commute.

Now we proceed to define the chain map $\calF$ from $X$ to $X_L$. This chain map will relate the codewords of $Q(X)$ to the codewords of $Q(X_L)$. It will also play an important role when we are studying the properties of the code $Q(X_L)$ in~\Cref{subsec:prop}.  For the convenience of describing the chain map, we would first define some regions in the faces $\Tilde{F}$:

\begin{itemize}
    \item $S$ contains vertices with $0\leq i,j\leq L-1$,
    \item $T$ contains vertices with $0\leq i\leq L-1,j=L$ or $0\leq j\leq L-1,i=L$.
    \item $U$ contains the vertice $i=j=L$.
\end{itemize}
Note that for the dummy faces, the regions would only include valid vertices, as shown in~\Cref{fig:squarediv}.

\begin{figure}
    \centering
    \begin{subfigure}[b]{0.3\textwidth}

        \centering
        \resizebox{\linewidth}{!}{
        \begin{tikzpicture}
        \foreach \row in {0,1,2}
        {
            \foreach \column in {0,1,2}
            {
                \pgfmathtruncatemacro{\rw}{2*\row}
                \pgfmathtruncatemacro{\cl}{2*\column}
                \node at (2*\row,2*\column) [circle, draw,blue](v\rw\cl){};
                \pgfmathtruncatemacro{\cl}{2*\column+1}
                \node at (2*\row,2*\column+1) [circle,fill,inner sep=1pt](v\rw\cl) {};
                \pgfmathtruncatemacro{\rw}{2*\row+1}
                \node at (2*\row+1,2*\column+1) [circle, draw,red](v\rw\cl){};
                \pgfmathtruncatemacro{\cl}{2*\column}
                \node at (2*\row+1,2*\column) [circle,fill,inner sep=1pt](v\rw\cl) {};
            };
        };

    \foreach \row in {0,1,...,5}
    {
        \foreach \column in {0,1,...,4}
        {
            \pgfmathtruncatemacro{\cl}{\column+1}
            \draw (v\row\column)--(v\row\cl);
        }
    }
    \foreach \column in {0,1,...,5}
    {
        \foreach \row in {0,1,...,4}
        {
            \pgfmathtruncatemacro{\rw}{\row+1}
            \draw (v\row\column)--(v\rw\column);
        }
    }
    \draw[blue!40] (-0.3,-0.3) rectangle (4.3,4.3);
    \draw[black!40] (4.7,-0.3) rectangle (5.3,4.3);
    \draw[black!40] (-0.3,4.7) rectangle (4.3,5.3);
    \draw[red!40] (4.7,4.7) rectangle (5.3,5.3);
    \node at (3,-0.5){$S$};
    \node at (5.5,3) {$T$};
    \node at (3,5.5) {$T$};
    \node at  (5.5,5.5){$U$};
    \end{tikzpicture}
        }
    \caption{Subdivision of a full face}
    \label{subfig:fullface}
    \end{subfigure}
    \begin{subfigure}[b]{0.3\textwidth}
     \resizebox{\linewidth}{!}{
     \begin{tikzpicture}
        \foreach \row in {0,1,2}
        {
            \foreach \column in {0,1,2}
            {
                \pgfmathtruncatemacro{\rw}{2*\row}
                \pgfmathtruncatemacro{\cl}{2*\column}
                \node at (2*\row,2*\column) [circle, draw,blue](v\rw\cl){};
                \pgfmathtruncatemacro{\cl}{2*\column+1}
                \ifnum \cl=5

                \else
                    \node at (2*\row,2*\column+1) [circle,fill,inner sep=1pt](v\rw\cl) {};
                    \pgfmathtruncatemacro{\rw}{2*\row+1}

                    \ifnum \rw=5

                    \else
                        \node at (2*\row+1,2*\column+1) [circle, draw,red](v\rw\cl){};
                        \pgfmathtruncatemacro{\cl}{2*\column}
                        \node at (2*\row+1,2*\column) [circle,fill,inner sep=1pt](v\rw\cl) {};
                    \fi
                \fi
            };
        };
    \node at (0,5) [circle,fill,inner sep=1pt](v05){};
    \node at (5,0) [circle,fill,inner sep=1pt](v50){};
    \node at (1,4) [circle,fill,inner sep=1pt](v14){};
    \node at (3,4) [circle,fill,inner sep=1pt](v34){};

    \foreach \row in {0,1,...,4}
    {
        \foreach \column in {0,1,...,3}
        {
            \pgfmathtruncatemacro{\cl}{\column+1}
            \draw (v\row\column)--(v\row\cl);
        }
    }
    \foreach \column in {0,1,...,4}
    {
        \foreach \row in {0,1,...,3}
        {
            \pgfmathtruncatemacro{\rw}{\row+1}
            \draw (v\row\column)--(v\rw\column);
        }
    }

    \draw (v04)--(v05);
    \draw (v40)--(v50);
    \draw[blue!40] (-0.3,-0.3) rectangle (4.3,4.3);
    \draw[black!40] (4.7,-0.3) rectangle (5.3,0.3);
    \draw[black!40] (-0.3,4.7) rectangle (0.3,5.3);
    \node at (3,-0.5){$S$};
    \node at (5.5,0.5) {$T$};
    \node at (0,5.5) {$T$};

    \end{tikzpicture}}
    \caption{Dummy face of level 0 vertex}
    \end{subfigure}
     \begin{subfigure}[b]{0.3\textwidth}
     \resizebox{\linewidth}{!}{\begin{tikzpicture}
        \foreach \row in {0,1,2}
        {
            \foreach \column in {0,1,2}
            {
                \pgfmathtruncatemacro{\rw}{2*\row}
                \pgfmathtruncatemacro{\cl}{2*\column}
                \node at (2*\row,2*\column) [circle, draw,blue](v\rw\cl){};
                \pgfmathtruncatemacro{\cl}{2*\column+1}
                \node at (2*\row,2*\column+1) [circle,fill,inner sep=1pt](v\rw\cl) {};
                \pgfmathtruncatemacro{\rw}{2*\row+1}
                \node at (2*\row+1,2*\column+1) [circle, draw,red](v\rw\cl){};
                \pgfmathtruncatemacro{\cl}{2*\column}
                \node at (2*\row+1,2*\column) [circle,fill,inner sep=1pt](v\rw\cl) {};
            };
        };

    \foreach \row in {1,2,...,5}
    {
        \foreach \column in {1,2,...,4}
        {
            \pgfmathtruncatemacro{\cl}{\column+1}
            \draw (v\row\column)--(v\row\cl);
        }
    }
    \foreach \column in {1,2,...,5}
    {
        \foreach \row in {1,2,...,4}
        {
            \pgfmathtruncatemacro{\rw}{\row+1}
            \draw (v\row\column)--(v\rw\column);
        }
    }
    \draw (v00)--(v01);
    \draw (v00)--(v10);
    \draw (v50)--(v51);
    \draw (v05)--(v15);
    \draw[blue!40] (0.7,0.7) rectangle (4.3,4.3);
    \draw[black!40] (4.7,-0.3) rectangle (5.3,4.3);
    \draw[black!40] (-0.3,4.7) rectangle (4.3,5.3);
    \draw[red!40] (4.7,4.7) rectangle (5.3,5.3);
    \filldraw[white] (-0.5,-0.5) rectangle (0.5,4.5);
    \filldraw[white] (-0.5,-0.5) rectangle (4.5,0.5);
    \node at (3,0.5){$S$};
    \node at (5.5,3) {$T$};
    \node at (3,5.5) {$T$};
    \node at  (5.5,5.5){$U$};
    \end{tikzpicture}}
    \caption{Dummy face of level 2 vertex}
    \end{subfigure}
    \caption{The subdivision of different faces,  where the blue vertices stand for level 0 vertices, black vertices stand for level 1 vertices, and red vertices stand for level 2 vertices.  We also labeled the image of chain map $\calF$ as $S,T,U$ regions.}
    \label{fig:squarediv}
\end{figure}
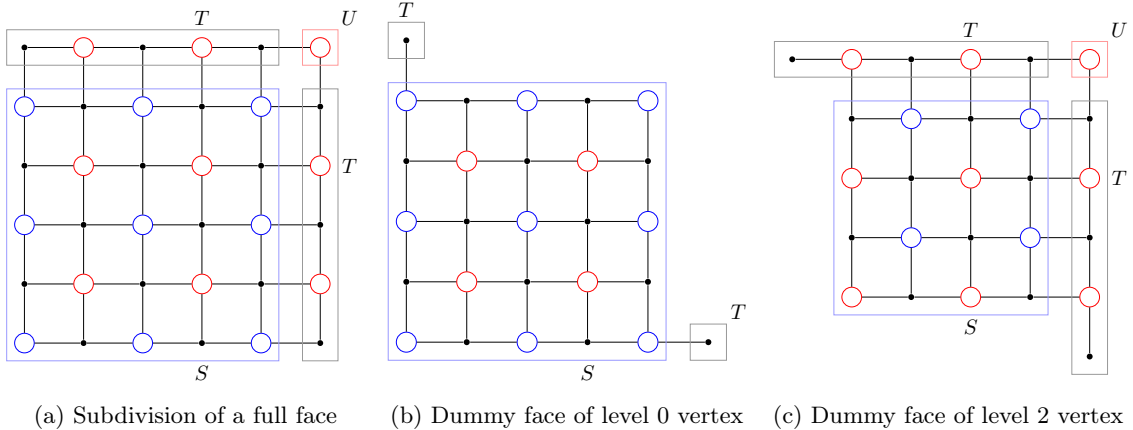

We define the maps $\calF_i\colon \F_2^{X(i)}\to \F_2^{X_L(i)},i\in\{0,1,2\}$:
\begin{itemize}
    \item Given $\Tilde{c}_0\in \F_2^{X(0)}$, we define $\calF_0(\Tilde{c}_0)$ by repeating the value $\Tilde{c}_0(v_0)$ at each component $S_{v_0}$ corresponding to $v_0\in X(0)$, including the level 0 dummy faces. For level 2 dummy faces, we set them to 0.
    \item Given $\Tilde{c}_1\in \F_2^{X(1)}$,  we define $\calF_1(\Tilde{c}_1)$ by repeating the value $\Tilde{c}_1(v_1)$ at each component $T_{v_1}$ corresponding to $v_1\in X(1)$, and the values in other part to $0$.
    \item Given $\Tilde{c}_2\in \F_2^{X(2)}$, we define $\calF_2(\Tilde{c}_2)$ by setting the value $\Tilde{c}_2(v_2)$ a the corresponding vertex $U_{v_2}$, and 0 elsewhere.
\end{itemize}

We have the following theorem for the maps $\calF_i$.
\begin{theorem}
    The maps $\calF_i(i=0,1,2)$ form a chain map between $X$ and $X_L$. Equivalently, we have the following commutation diagram:
    \[
    \begin{tikzcd}
    \F_2^{X(0)}\arrow[r,"\delta_0"]\arrow[d,"\calF_0"] &\F_2^{X(1)}\arrow[r,"\delta_1"] \arrow[d,"\calF_1"]&\F_2^{X(2)}\arrow[d,"\calF_2"]\\
    \F_2^{X_L(0)}\arrow[r,"\delta_0"] &\F_2^{X_L(1)}\arrow[r,"\delta_1"] &\F_2^{X_L(2)}
    \end{tikzcd}
\]

\end{theorem}

\begin{proof}
    By our construction of $\calF$, we can see that inside each region $S$, $\delta_0\calF_0(c)|_S=0$. Since $\calF_1(\delta_0 c)$ would only have nonzero images in $T$, we can check that the commutation diagram holds inside each region $S$.

We proceed to check the commutation relation between different regions $S,T,U$.  Consider the neighbor structure of the regions in $X_L$. Note that for the regions we specified in~\Cref{fig:squarediv}, the $\delta$ operator would be applied in the direction $S\to T\to U$.

We would first verify the commutation relation for the coboundary operators $\delta$ between $U$ and $T$, i.e. $\delta_1\calF_1=\calF_2 \delta_1$. For each region $U$, consider its preimage $u$ in $X(2)$, we can see that for each $v_1\in X(1)$ that is connected to $v$, we have exactly one region $T_{v_1}$ in $X_L$ that is connected to $U$. Since the level 1 vertex that connects to $U$ in $T_v$ will be assigned value $c_1(v)$, this implies that $\delta_1\calF_1=\calF_2 \delta_1$ for the part of $\delta_1$ implied by the edges between $T$ and $U$.

Now we verify the commutation relation for coboundary operators $\delta$ implied by the edges between the region $S$ and $T$. We first note that all the level 1 vertices in $S$ and level 2 vertices in $T$ are assigned to 0, and each level 2 vertex in $T$ connects to two level 1 vertices in $T$ with the same value, thus we have $\delta_1\calF_1=\calF_2 \delta_1$. We can observe that the region $T$ has the structure as shown in~\Cref{fig:Tstruct}, which has a center level 1 vertex and several branches. For the center level 1 vertex of $T$, we observe that each $\delta_0$ edge from a level 0 vertex connected to it would also correspond to an original $\delta_0$ edge in $X$. When we are considering the branches of $T$, we can ignore the dummy faces, since the level 0 dummy faces are not connected to $T$, while level 2 dummy faces would be assigned to all zero. For each branch in $T$, we can observe that it corresponds to an edge from a level 1 vertex to a level 2 vertex in $X$. From our construction, each level 0 vertex $v_0\in X(0)$ would have exactly one face $f$ in $S_{v_0}$ that is connected to this branch, and the level 0 vertices that connect to level 1 vertices in this branch are all assigned as $c_0(v_0)$. Since each level 1 vertex are assigned as $c_1(v_1)=(\delta_0 c_0)(v_1)$ we have $\delta_0\calF_0=\calF_1\delta_0$.
\end{proof}

\begin{figure}
    \centering
    \begin{subfigure}[b]{0.5\textwidth}
        \resizebox{\linewidth}{!}{
        \begin{tikzpicture}
            \foreach \row in {0,1,2}
            {
                \node at (0,2*\row) [circle,fill,inner sep=1.5pt] (v0\row){};
                \node at (0,2*\row+1) [circle, draw,red] (u0\row){};
                \node at (0,-2*\row) [circle,fill,inner sep=1.5pt] (-v0\row){};
                \node at (0,-2*\row-1) [circle, draw,red] (-u0\row){};
                \node at (2*\row,0) [circle,fill,inner sep=1.5pt] (w\row0){};
                 \node at (2*\row+1,0) [circle, draw,red] (x\row0){};
                \node at (-2*\row,0) [circle,fill,inner sep=1.5pt] (-w\row0){};
                \node at (-2*\row-1,0) [circle, draw,red] (-x\row0){};
            }
            \draw (v00)--(u00);
            \draw (-v00)--(-u00);
            \draw (w00)--(x00);
            \draw (-w00)--(-x00);
            \draw (u00)--(v01);
            \draw (-u00)--(-v01);
            \draw (x00)--(w10);
            \draw (-x00)--(-w10);
            \draw (v01)--(u01);
            \draw (-v01)--(-u01);
            \draw (w10)--(x10);
            \draw (-w10)--(-x10);
            \draw (u01)--(v02);
            \draw (-u01)--(-v02);
            \draw (x10)--(w20);
            \draw (-x10)--(-w20);
            \node at (0,5) [circle,fill, white, inner sep=5pt] {};
            \node at (0,-5) [circle,fill, white, inner sep=5pt] {};
            \node at (5,0) [circle,fill, white, inner sep=5pt] {};
            \node at (-5,0) [circle,fill, white, inner sep=5pt] {};
            % \draw[dashed] (v02)--(u02);
            % \draw[dashed] (-v02)--(-u02);
            % \draw[dashed] (w20)--(x20);
            % \draw[dashed] (-w20)--(-x20);
       \end{tikzpicture}
        }
    \end{subfigure}
    \label{fig:Tstruct}
    \caption{An example structure of the $T$ region, where we use black nodes for level 1 vertices and red nodes for level 2 vertices.}
\end{figure}
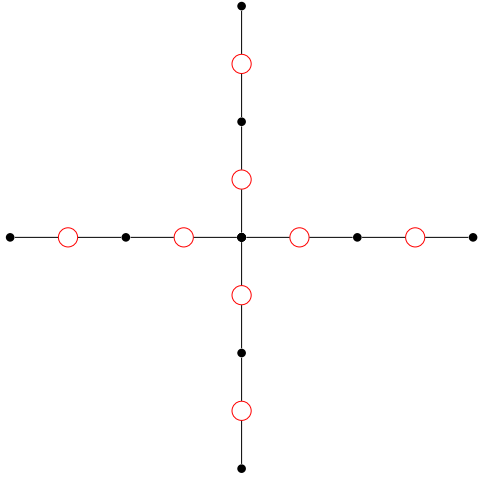

\subsection{Geometrically Local Embedding}
\label{subsec:embed}
In~\Cref{subsec:divqldpc}, we have constructed the code $Q(X_L)$ together with its corresponding 2D subspace complex $\Tilde{\calS}_L(Q)$. We can construct the map $I_{\code_L\to\text{square}_L}$ by mapping the chain complex $X_L$ to its corresponding 2D subspace complex $\Tilde{\calS}_L(Q)$, identifying $\Tilde{V}_L=X_L(0)\sqcup X_L(1)\sqcup X_L(2)$, and edges $\Tilde{E}$ with nonzero entries in (co)boundary operators.

In this section, we will show that the code $Q(X_L)$ has an geometrically local embedding. We want to apply the map from~\Cref{thm:embedding} to obtain the embedding of the code. Note that while $\Tilde{\calS}$ is not strictly a square complex, we can apply~\Cref{thm:embedding} to its underlying complex, i.e. its downward closed completion.
Let ${\calS}_u(Q)=(V_u,E_u,F_u)$ be the downward closed completion of $\Tilde{\calS}(Q)$
    and $\calS_L(Q)=(V_L,E_L,F_L)$ be its subdivision with parameter $L$.
It is easy to see that $\Tilde{\calS}_L(Q)=(\Tilde{V}_L,\Tilde{E}_L,\Tilde{F}_L)$ is a subspace of $\calS_L(Q)$.

If we set $L=\Theta(|V_{u}|^{\frac{1}{D-2}})$ for the map $I_{\text{square}_L\to \text{euclid}}$ from~\Cref{thm:embedding}, we can obtain a geometrically local and bounded density embedding for $V_L$ in $\Z^D$. Now we verify that $I_{\text{square}_L\to \text{euclid}}|_{\Tilde{V}_L}$ is a geometrically and bounded density local map. It is easy to verify that the bounded density property still holds, since $\Tilde{V}_L$ is a subset of $V_L$, thus for any $x$, its preimage won't increase. For the geometrically local property, since $\Tilde{E}_L\subset E_L$, thus for any neighbor $(v_0,v_1)\in \Tilde{E}_L$, $|f(v_0)-f(v_1)|\leq a$.

By embedding the vertices $\Tilde{V}_L$ in $\Z^D$, we also obtained a geometrically local embedding for vertices in $X_L(i)$, since $\Tilde{V}_L=X_L(0)\sqcup X_L(1)\sqcup X_L(2)$, thus the code $Q(X_L)$ indeed have a geometrically local embedding $I_{\code_L\to\euclid}=I_{\text{square}_L\to \text{euclid}}\circ I_{\code_L\to\text{square}_L}$.

\section{Properties of the Subdivided Code}
\label{subsec:prop}

In this section, we will introduce some results from ~\cite{linGeometricallyLocalQuantum2023} that will provide us the bounds on the dimension, distance, and energy barrier of our code $Q(X_L)$. Though our construction is different from their construction, we observe that some of their proofs can be applied more generally to our setting. Thus in this section, we will first introduce some theorems from~\cite{linGeometricallyLocalQuantum2023}, and discuss the changes that should be made in the proof of our construction. Interested readers can refer to their paper for the full proof.

We first summarize the properties of our code in the following theorem:

\begin{theorem}[Formal Statement of~\Cref{thm:main}]\label{thm:formal}
    Given a quantum LDPC code $Q$ on $n$ qubits, with dimension $k=\Theta(n)$ and distance $d=\Theta(n)$, for any dimension $D\geq 3$, there exists a $D$ dimension geometrically local subdivided code $Q_L$ on $N=\Theta(n^{\frac{D}{D-2}})$ qubits, with dimension $k=\Omega(N^{\frac{D-2}{D}})$, distance $d=\Omega(N^{\frac{D-1}{D}})$. Moreover, if the quantum LDPC code has energy barrier $\calE=\Theta(n)$, our new code has energy barrier $\calE=\Omega(N^{\frac{D-2}{D}})$.
\end{theorem}

  We will define some expansion properties of the complex $X_L$ in this section, and use these properties to show the bounds of the code dimension, distance, and energy barrier. We will prove these expansion properties in~\Cref{sec:exp}.

The following estimation on the size of our new chain complex $X_L$ turns out to be useful:
\begin{claim}\label{clm:size}
  Given the maximum degree $\Delta$ of the chain complex $X$, we have the following bound for the complex $X_L$
 \begin{align*}
      \frac{L^2}{4}|X(i)|\leq |X_L(i)|\leq \frac{\Delta^2L^2}{2}(|X(0)|+2|X(i)|+|X(2)|).
 \end{align*}
 In particular, if $\Delta = \Theta(1)$, we have $|X_L(i)|=\Theta(L^2|X(i)|)$.
\end{claim}
This claim can be verified by counting the size of $\Tilde{F}$ of our square subspace complex $\Tilde{\calS}(Q)$. Recall that $\Tilde{F}$ consists of two parts, the full faces $F$ that we obtained by pairing qubits and the dummy faces $F_D$ we added for local edge connectivity. It is easy to see that $|F_D|$ is bounded by $\frac{\Delta^2}{2}(|X(0)|+|X(2)|)$, since for each $v\in X(0)$ or $X(2)$, we can add one dummy face between every pair of edges. To bound the size of the normal face set $|F|$, we observe that for each path $v_0\to v_1\to v_2$, where $v_i\in X(i)$ and is connected in the graph $\calT(Q)$, they can determine exactly one normal face in $F$ by our pairing algorithm. Since there is at most $\Delta$ edges from a given vertex $v_0$ to $X(1)$, and $\Delta$ edges for each $v_1$ to $X(2)$, thus for each $v_0\in X(0)$ it can have at most $\Delta^2$ neighboring full faces. Similar arguments also apply to vertices in $X(1)$ and $X(2)$. We can see that after the subdivision, the number of level-$i$ vertices in $X_L$ is increased at most $(L+1)^2/4$ times, giving our upper bound.

For general good qLDPC codes, we will have that $\Delta$ is a constant, and $|X(i)|$ are of the same order ($\frac{1}{\Delta} |X(j)| \le |X(i)| \le \Delta |X(j)|$ for any $i, j \in \{0, 1, 2\}$). Thus, we have $|X_L(i)|=\Theta(L^2|X(i)|)$.

\begin{remark}
    Note that if our code is reasonable, that is we do not have the dummy face set $F_D$, we can directly obtain that $|X_L(i)|\leq \Delta^2L^2|X(i)|$.
    We can also observe that for each $v$, its face neighborhood size $|N^F(v)|$ is bounded by $\Delta^2+\Delta^2/2=\frac{3\Delta^2}{2}$.

\end{remark}

For the dimension of the code $Q(X_L)$, we have the following lemma,
\begin{lemma}[Lemma 4.2 in~\cite{linGeometricallyLocalQuantum2023}]\label{lem:dim}
    If the code $Q(X)$ has dimension $k$, then the code $Q(X_L)$ has dimension $k$.
\end{lemma}

To prove this lemma, we have to show that $\calF_1$ induces a bijection between the codewords (equivalent classes) $Z^1(X)/B^1(X)$ and $Z^1(X_L)/B^1(X_L)$ by mapping $[\Tilde{c}_1]$ to $[\calF_1(\Tilde{c}_1)]$. The proof of the theorem mainly used the fact that $\calF$ and $\delta$ commute to help us find images/preimages of cocycles and coboundaries. When it comes to finding the preimage under $\calF_1$ of a cocycle $c_1\in Z^1(X_L)$, we would also use that $X_L|_S$ is exact, which can be verified in our construction.

To obtain the distance and energy barrier lower bound of the code $Q(X_L)$, we would use the coboundary expansion properties of the chain complex $X_L$. It is helpful to introduce the definition of small-set coboundary expansion from the high dimensional expander literature.

\begin{definition} [Small-Set (Co)Boundary Expansion]
  We say that $X: \F_2^{X(0)} \xrightarrow{\delta_0} \F_2^{X(1)} \xrightarrow{\delta_1} \F_2^{X(2)}$ is a $(\alpha,\beta,\gamma)$-small-set boundary expander if
  \begin{align*}
    \forall c_1 \in \F_2^{X(1)}, \abs{c_1} \le \alpha |X(1)|:
    \exists c_2 \in \F_2^{X(2)}, \beta \abs{c_1 + \partial_2 c_2} \le \abs{\partial_1 c_1}, \gamma \abs{c_2} \le \abs{c_1}.
  \end{align*}

  Similarly, $X$ is a $(\alpha,\beta,\gamma)$-small-set coboundary expander if
  \begin{align*}
    \forall c_1 \in \F_2^{X(1)}, \abs{c_1} \le \alpha |X(1)|:
    \exists c_0 \in \F_2^{X(0)}, \beta \abs{c_1 + \delta_0 c_0} \le \abs{\delta_1 c_1}, \gamma \abs{c_0} \le \abs{c_1}.
  \end{align*}
\end{definition}
By the symmetry between the $X$ and $Z$ vertices in our construction, in the following parts, we would only prove the coboundary expansion properties of the code, as the boundary expansion proof would be the same.

It is known that we can relate the distance and energy barrier of a code $Q(X_L)$ to the small set (co)boundary expansion parameters of its corresponding complex $X_L$ through the following theorem:

\begin{theorem}[Implict in Section 2.4 of~\cite{linGeometricallyLocalQuantum2023}]\label{thm:dis}
If the complex $X_L$ has $(\alpha,\beta,\gamma)$-small-set boundary and coboundary expansion, the corresponding code has distance $d\geq \alpha|X_L(1)| $ and energy barrier $\calE\geq \alpha\beta|X_L(1)|$.
\end{theorem}
The proof of the theorem is merely a translation between languages of coding theory and expanders. Thus the converse statement is also true. That is, if we have a qLDPC code $Q(X)$ with distance $d=\Omega(n)$ and energy barrier $\calE=\Omega(n)$, the corresponding complex $X$ has $(\alpha,\beta,\gamma)$-small-set (co)boundary expansion, where $\alpha,\beta=\Theta(1)$.

To obtain an optimal lower bound of the code distance and energy barrier, we would try to prove the complex $X_L$ has $(\alpha,\beta,\gamma)$-small-set coboundary expansion, where $\alpha,\beta,\gamma=\Theta(1/L)$. However, directly proving the coboundary expansion bounds turns out to be difficult. It is shown in~\cite{linGeometricallyLocalQuantum2023} that we can first prove local coboundary expansion properties, and show the global small-set coboundary expansion properties by `cleaning' the errors inside the surfaces and `pushing' the rest of them to the boundary.

To formally describe the `cleaning' and `pushing' process, we would use the definition of a chain complex with a boundary. The chain complex $Y\colon \F_2^{Y(0)} \xrightarrow{\delta_0} \F_2^{Y(1)} \xrightarrow{\delta_{1}} \F_2^{Y(2)}$ is extended to have the additional boundary structure $Y_{\partial}$, $Y\cup Y_{\partial}\colon\F_2^{Y(0) \cup Y_\partial(0)} \xrightarrow{\delta_0'} \F_2^{Y(1) \cup Y_\partial(1)} \xrightarrow{\delta_{1}'} \F_2^{Y(2) \cup Y_\partial(2)}$.\footnote{Note that in the examples we study, $\delta_{1}\delta_0 = 0$ but $\delta_{1}'\delta_0' \ne 0$. The maps $\delta_i'$ do not, and do not need to, form a chain complex.}

Let us first define the following coboundary expansion properties of the local chain complexes with boundaries:
\begin{definition}
    A chain complex $Y$ with boundary $Y_{\partial}$ is a $(\beta_i,\eta_i)$-coboundary expander at level $i$ if for all $\hat{f}_i\in\F_2^{Y(i)}$, there exists $f_i\in \hat{f}_i+B^i\subset \F_2^{Y(i)}$ such that
    \begin{align*}
    (1)\,|f_i|_{\intt}\leq |\hat{f}_i|_{\intt},\quad (2)\,\beta_i|f_i|_{\intt}\leq |\delta_i\hat{f}_i|_{\intt},\quad (3)\,\eta_i|\delta_i f_i|_{\partial}\leq |\delta_i f_i|_{\intt}.
\end{align*}
where $|\cdot|_{\intt}$ is the Hamming weight of the vector restricted to the space $\F_2^{Y(i)}$, and $|\cdot|_{\partial}$ is the Hamming weight of the vector restricted to the space $\F_2^{Y_{\partial}(i)}$.
\end{definition}
The inequality (2) is the standard coboundary expansion of the chain complex $Y$, and the inequality (3) is introduced for the cleaning process.

For our regions $S$ and $T$, they can be viewed as a chain complex with boundary $S_{\partial}$ and $T_{\partial}$, where  $S_{\partial}$ and $T_{\partial}$ correspond to vertices in the connected regions $T$ and $U$ respectively. The coboundary operators would include edges between $S$ and $S_{\partial}$, and $T$ and $T_{\partial}$. Their structure can be seen in~\Cref{fig:boundary}.

Note that $T$ would have the structure of a generalized repetition code in~\cite{linGeometricallyLocalQuantum2023}. From the Tanner graph perspective, the repetition code can be represented as a 1D chain of bits connected by a 1D chain of check vertices which require the neighboring bits to be the same.
The generalized repetition code is similar,
  except that now there could be one branching check at the center. We provide an example in~\Cref{fig:genrep}, where we use level 1 vertices as checks and level 2 vertices as bits for consistency with the $T$ structure. A generalized repetition code is said to have length $L$ if there are $L$ level 1 vertices on the path from one boundary to another boundary.
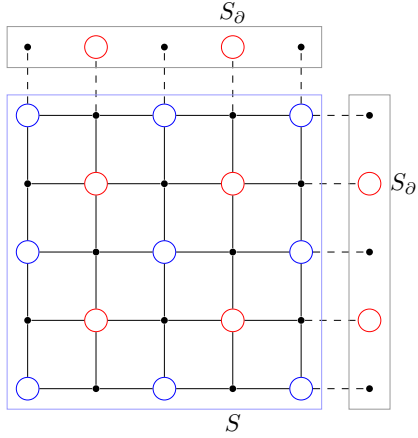
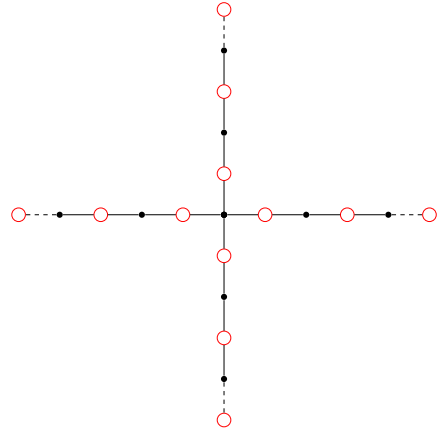
\begin{figure}
    \centering
    \begin{subfigure}[b]{0.35\textwidth}
        \centering
        \resizebox{\linewidth}{!}{
        \begin{tikzpicture}
        \foreach \row in {0,1,2}
        {
            \foreach \column in {0,1,2}
            {
                \pgfmathtruncatemacro{\rw}{2*\row}
                \pgfmathtruncatemacro{\cl}{2*\column}
                \node at (2*\row,2*\column) [circle, draw,blue](v\rw\cl){};
                \pgfmathtruncatemacro{\cl}{2*\column+1}
                \node at (2*\row,2*\column+1) [circle,fill,inner sep=1pt](v\rw\cl) {};
                \pgfmathtruncatemacro{\rw}{2*\row+1}
                \node at (2*\row+1,2*\column+1) [circle, draw,red](v\rw\cl){};
                \pgfmathtruncatemacro{\cl}{2*\column}
                \node at (2*\row+1,2*\column) [circle,fill,inner sep=1pt](v\rw\cl) {};
            };
        };

    \foreach \row in {0,1,...,4}
    {
        \foreach \column in {0,1,...,3}
        {
            \pgfmathtruncatemacro{\cl}{\column+1}
            \draw (v\row\column)--(v\row\cl);
        }
    }
    \foreach \column in {0,1,...,4}
    {
        \foreach \row in {0,1,...,3}
        {
            \pgfmathtruncatemacro{\rw}{\row+1}
            \draw (v\row\column)--(v\rw\column);
        }
    }
    \foreach \column in {0,1,...,4}
    {
        \draw[dashed] (v4\column)--(v5\column);
        \draw[dashed] (v\column4)--(v\column5);
    }
    \draw[blue!40] (-0.3,-0.3) rectangle (4.3,4.3);
    \draw[black!40] (4.7,-0.3) rectangle (5.3,4.3);
    \draw[black!40] (-0.3,4.7) rectangle (4.3,5.3);
    \filldraw[white] (4.5,4.5) rectangle (5.5,5.5);
    \node at (3,-0.5){$S$};
    \node at (5.5,3) {$S_{\partial}$};
    \node at (3,5.5) {$S_{\partial}$};
    \end{tikzpicture}
        }
    \caption{the figure of one face of $S$ with boundary $S_{\partial}$}
    \end{subfigure}
    \qquad\qquad\qquad
    \begin{subfigure}[b]{0.35\textwidth}
        \centering
        \resizebox{\linewidth}{!}{
        \begin{tikzpicture}
            \foreach \row in {0,1,2}
            {
                \node at (0,2*\row) [circle,fill,inner sep=1.5pt] (v0\row){};
                \node at (0,2*\row+1) [circle, draw,red] (u0\row){};
                \node at (0,-2*\row) [circle,fill,inner sep=1.5pt] (-v0\row){};
                \node at (0,-2*\row-1) [circle, draw,red] (-u0\row){};
                \node at (2*\row,0) [circle,fill,inner sep=1.5pt] (w\row0){};
                 \node at (2*\row+1,0) [circle, draw,red] (x\row0){};
                \node at (-2*\row,0) [circle,fill,inner sep=1.5pt] (-w\row0){};
                \node at (-2*\row-1,0) [circle, draw,red] (-x\row0){};
            }
            \draw (v00)--(u00);
            \draw (-v00)--(-u00);
            \draw (w00)--(x00);
            \draw (-w00)--(-x00);
            \draw (u00)--(v01);
            \draw (-u00)--(-v01);
            \draw (x00)--(w10);
            \draw (-x00)--(-w10);
            \draw (v01)--(u01);
            \draw (-v01)--(-u01);
            \draw (w10)--(x10);
            \draw (-w10)--(-x10);
            \draw (u01)--(v02);
            \draw (-u01)--(-v02);
            \draw (x10)--(w20);
            \draw (-x10)--(-w20);
             \draw[dashed] (v02)--(u02);
            \draw[dashed] (-v02)--(-u02);
            \draw[dashed] (w20)--(x20);
            \draw[dashed] (-w20)--(-x20);
       \end{tikzpicture}
        }
        \caption{Example of $T$ and $T_{\partial}$ with $L=5$}
        \label{fig:genrep}
    \end{subfigure}
    \caption{Examples of $S$ and $T$ with boundaries, and the level 0, 1, 2 vertices follows our coloring convention. Note that, since the dummy faces only have `imaginary edges', they do not contribute to the boundary $S_\partial$.}
    \label{fig:boundary}
\end{figure}

We can relate local expansion to global expansion through the following theorem:
\begin{theorem}[Theorem 4.4 in~\cite{linGeometricallyLocalQuantum2023}]
    If $X$ has $(\alpha_{\mathrm{qLDPC}},\beta_{\mathrm{qLDPC}},\gamma_{\mathrm{qLDPC}})$-small-set coboundary expansion with $\alpha_{\mathrm{qLDPC}},\beta_{\mathrm{qLDPC}},\gamma_{\mathrm{qLDPC}}=\Theta(1)$, and for each local complex $S$ and $T$, they are $(\beta_i,\eta_i)$-coboundary expanders with $\beta_i=\Theta(1/L),\eta_i=\Theta(1)$(for $S$ we take $i=0,1$, and for $T$ we only consider $i=0$), we have that the complex $X_L$ has $(\alpha,\beta,\gamma)$-small-set coboundary expansion with $\alpha,\beta,\gamma=\Theta(1/L)$.
\end{theorem}

To prove this theorem, we would also relate $c_1\in \F_2^{X_L(1)}$ back to some codeword in $X(1)$. We briefly review the cleaning process for the special case when $c_1$ is a cocycle and have cosystolic distance $|c_1|\leq\alpha |X_L(1)|$, where $\alpha=\Theta((\eta_0^S/\Delta^2L)\alpha_{\mathrm{qLDPC}})$. In our parameter setting, we can obtain that $\alpha=\Theta(1/L)$. For our special case, we can prove that $c_1\in B^1(X_L)$, i.e. it is a coboundary. We would first remove the error inside each region $S$ using some coboundary  $\delta_0 c_0^S$ by the local expansion property of $S$. We can check that the codeword in $T$ is the image of some $\Tilde{c}_1\in Z^1(X)$ under $\calF_1$.  The size of $\Tilde{c}_1$ can be bounded by our estimation of size $|X_L(1)|$ and the coboundary expansion of $S$, thus we can prove $\Tilde{c}_1\in B^1(X)$ by the coboundary expansion properties of $X$. Finally, since $\calF$ commutes with $\delta$. we proved $c_1\in B^1(X_L)$. The full proof for general case is similar, and we need extra cleaning for $S$ and $T$ using their local expansion properties.

We will prove the local chain complexes $S$ and $T$ are good coboundary expanders in~\Cref{sec:exp}. Note that the local chain complexes are different from~\cite{linGeometricallyLocalQuantum2023}. We proved similar expansion bounds for our modified local complex $S$, without using the product structure in their proof, and our proof can be generalized to any edge-connected local complex.

We summarize our result in the following two theorems:
\begin{theorem}
    \label{thm:level0}
    Every local complex $S$ with the boundary $ S_\partial$ is a $(\beta^S_0,\eta^S_0)$-coboundary expander at level 0, and every local complex $T$ with boundary $ T_\partial$ is a $(\beta^T_0,\eta^T_0)$-coboundary expander at level 0. We have that $\beta_0^S,\beta_0^T=\Theta(1/L)$, $\eta_0^S,\eta_0^T=\Theta(1)$.
\end{theorem}

\begin{theorem}
\label{thm:level1}
     Every local complex $S$ with the boundary $ S_\partial$ is a $(\beta^S_1,\eta^S_1)$-coboundary expander at level 1. We have that $\beta_1^S=\Theta(1/L)$, $\eta_1^S=\Theta(1)$.
\end{theorem}

We are ready to prove~\Cref{thm:formal}.
\begin{proof}[Proof of~\Cref{thm:formal}]
    % \david{I think the notation $G'$ is not used, need to update}
     Recall in~\Cref{subsec:embed}, we set $L=\Theta(|V_{u}|^{\frac{1}{D-2}})$. It is easy to verify that $|V_{u}|=\Theta(|X(0)|+|X(1)|+|X(2)|)$ since we only add at most $\Delta^2/2$ of dummy faces to the structure for each check, thus $|X(i)|=\Theta(L^{D-2})$. By Claim~\ref{clm:size}, we have $|X_L(i)|=\Theta(L^{D})$. We would apply the map  $I_{\text{square}_L\to \text{euclid}}$ from~\Cref{thm:embedding} to our subdivided complex $X_L$ to obtain a geometrically local code.
    Let $N=|X_L(1)|$, by Lemma~\ref{lem:dim} and~\Cref{thm:dis}, we have that the corresponding code has dimension $k=\Omega(N^{\frac{D-2}{D}})$, distance $d=\Omega(N^{\frac{D-1}{D}})$, and energy barrier $\calE=\Omega(N^{\frac{D-2}{D}})$.
\end{proof}

\section{Expansion Analysis of local structures }\label{sec:exp}
In this section, we prove that every $S$ is a $(\beta_0=\Theta(1/L),\eta_0=\Theta(1))$-coboundary expander at level 0, $(\beta_1=\Theta(1/L),\eta_1=\Theta(1))$-coboundary expander at level 1, and every $T$ is a $(\beta_0=\Theta(1/L),\eta_0=\Theta(1))$-coboundary expander at level 0.

\subsection{Proof of~\Cref{thm:level0}}
 We would prove the level 0 expansion properties of $S$ and $T$ using functional inequalities. To be specific, we will study the functional inequality on a graph with boundary $G=(V, E, V^\partial)$. Besides the vertex set $V$ and the edge set $E\subset V\times V$, the graph has an additional boundary vertex set $V^\partial$.
 Let us first recall the definition of functional inequalities from~\cite{linGeometricallyLocalQuantum2023}.
\begin{definition}
We say a graph with boundary $G=(V,E,V^\partial)$ satisfies $(C,C^\partial)$-functional inequalities if for all functions $g\colon V\cup V^\partial\to \zo$:
\begin{align*}
    \sum_{\{x,y\}\in E}|g(x)-g(y)|&\geq\frac{C}{|V|} \sum_{x,y\in V}|g(x)-g(y)|,\\
    \sum_{\{x,y\}\in E}|g(x)-g(y)|&\geq\frac{C^\partial}{|V^\partial|} \sum_{x\in V^\partial,y\in V}|g(x)-g(y)|.
\end{align*}
\end{definition}

We can obtain the level 0 expansion properties of a graph through the following claim.
\begin{claim}[Claim 5.10 in~\cite{linGeometricallyLocalQuantum2023}]
\label{clm:fineq}
    If a graph with boundary satisfies the $(C,C^\partial)$ functional inequality, the corresponding complex is a $\left(\beta_0=C,\eta_0=\frac{C^\partial|V|}{2|V^\partial|}\right)$-coboundary expander at level 0.
\end{claim}

Since the $T$ components are still generalized repetition codes in our new construction, its expansion results and functional inequalities still apply. We summarize the results in~\cite{linGeometricallyLocalQuantum2023} here.

\begin{lemma}[Lemma 5.2 and Claim 5.9 in~\cite{linGeometricallyLocalQuantum2023}]\label{lem:repcode}
    The graph of generalized repetition code satisfies $(C^{rep}=1/L,C^{rep,\partial}=1/L)$-functional inequalities and the code is a $(\beta_0=2/L,\eta_0=1)$-coboundary expander at level 0.
\end{lemma}

Note that we would also consider the case when the generalized repetition code only has boundary vertices on some of its ends. It is easy to check that Lemma~\ref{lem:repcode} still holds in this case.

We provide an example of the local structure $S$ in~\Cref{fig:seam}. Please note that in this section, for better exposition of the functional inequalities, the vertices, edges, and faces in the figures will represent the vertices in $X(0)$, $X(1)$, $X(2)$ respectively.
We would call the subgraph whose edges branch from the center of $S$ the seam of $S$, as the $\calM$ part in~\Cref{fig:seam}. For the local complex $S$ in the new construction, we state its difference with the generalized face code as follows:
\begin{itemize}
    \item The new component $S$ does not have a product structure of two repetition codes, there is a subdivided face between any two branches on the seam $\calM$.
    \item Between two branches on the seam, there could be multiple subdivided faces.
\end{itemize}

\begin{figure}
    \centering
      \begin{subfigure}[b]{0.4\textwidth}

        \centering
        \resizebox{\linewidth}{!}{
        \begin{tikzpicture}
    \def\size{5}
    \foreach \row in {0,1,...,\size}
    {
        \draw ({\row*cos(90)},{\row*sin(90)})--({\size*cos(330)+\row*cos(90)},{\size*sin(330)+\row*sin(90)});
        \draw ({\row*cos(90)},{\row*sin(90)})--({\size*cos(210)+\row*cos(90)},{\size*sin(210)+\row*sin(90)});
        \draw ({\row*cos(210)},{\row*sin(210)})--({\size*cos(90)+\row*cos(210)},{\size*sin(90)+\row*sin(210)});
        \draw ({\row*cos(210)},{\row*sin(210)})--({\size*cos(330)+\row*cos(210)},{\size*sin(330)+\row*sin(210)});
        \draw ({\row*cos(330)},{\row*sin(330)})--({\size*cos(90)+\row*cos(330)},{\size*sin(90)+\row*sin(330)});
        \draw ({\row*cos(330)},{\row*sin(330)})--({\size*cos(210)+\row*cos(330)},{\size*sin(210)+\row*sin(330)});
    }
    \draw[line width=2pt,color=magenta] (0,0)--({\size*cos(90)},{\size*sin(90)});
    \draw[line width=2pt,color=magenta] (0,0)--({\size*cos(210)},{\size*sin(210)});
    \draw[line width=2pt,color=magenta] (0,0)--({\size*cos(330)},{\size*sin(330)});
    \draw[line width=2pt,color=orange] (0,2)--({\size*cos(210)},{\size*sin(210)+2});
    \draw[line width=2pt,color=orange] (0,2)--({\size*cos(330)},{\size*sin(330)+2});
    \node at (0,2) [fill, circle, inner sep=3pt]{};
    \node at (0.4,2.3) {\Large $w$};
    \node at (5.2,-0.5) {\Large $D(w)$};
    \node at (5,-2.5) {\Large \textcolor{magenta}{$\calM$}};

\end{tikzpicture}

        }
    \end{subfigure}
    \caption{An example of $S$ with three faces (without drawing the boundary $S_{\partial}$). We call the magenta part as the seam of $S$, denoted by $\calM$. We also marked one generalized repetition code $D(w)$ under consideration in orange, with center $w$. We would use vertices, edges, and faces in the graph for level 0, 1, 2 vertices of $S$ respectively.}
    \label{fig:seam}
\end{figure}
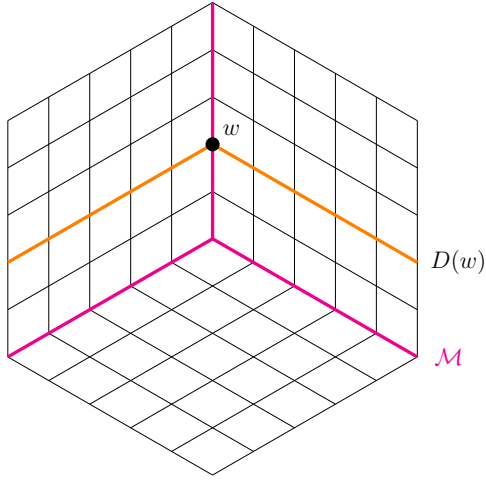

 When proving the level 0 expansion of the generalized face code, \cite{linGeometricallyLocalQuantum2023} utilized its product structure. In the following section, we will generalize the result to our local complex $S$. Note that the complex $S$ together with its boundary $S_{\partial}$ naturally induces a graph with boundary. 
 \begin{lemma}\label{lem:funineq}
     The graph with boundary induced by $S$ satisfies $(C=\Theta(1/L), C^{\partial}=\Theta(1/L))$ functional inequality.
 \end{lemma}

 As shown by Claim~\ref{clm:fineq}, if we prove the lemma beyond, we can show that $S$ is a $(\beta_0=\Theta(1/L),\eta_0=\Theta(1))$-coboundary expander at level 0. To simplify the formulae, we will use $L'=\frac{L+1}{2}$ in this section.
\begin{proof}[Proof of Lemma~\ref{lem:funineq}]
To prove the functional inequality of $S$, i.e.
\begin{align*}
     \sum_{\{x,y\}\in E}|g(x)-g(y)|&\geq\frac{C}{|V|} \sum_{x,y\in V}|g(x)-g(y)|,
\end{align*}
we first show
\begin{align}\label{eq:1}
    \sum_{\{x,y\}\in E}|g(x)-g(y)|&\geq \sum_{w\in \calM}\frac{C^{rep}}{\Delta L'}\sum_{x,y\in D(w)}|g(x)-g(y)|,
\end{align}
then show
\begin{align}\label{eq:2}
    \sum_{w\in\calM}\sum_{x,y\in D(w)}|g(x)-g(y)|&\geq \frac{1}{|N^F(v)|^2L'}\sum_{x,y\in V}|g(x)-g(y)|,
\end{align}
where we recall that $N^F(v)$ is the set of faces containing $v$.

Combining the two inequalities together, we obtain
\begin{align*}
     \sum_{\{x,y\}\in E}|g(x)-g(y)|&\geq \frac{C^{rep}}{\Delta L'} \sum_{w\in \calM}\sum_{x,y\in D(w)}|g(x)-g(y)|\\
     &\geq  \frac{C^{rep}}{|N^F(v)|^2\Delta L'^2} \sum_{x,y\in V}|g(x)-g(y)|.
\end{align*}
Since $|N^F(v)|L'^2\leq |V|\leq |N^F(v)|(L'+1)^2\leq |N^F(v)|L^2$, we can observe that if we set $C=C^{rep}/2\Delta^3$,
\begin{align*}
    \frac{C}{|V|}\leq \frac{C^{rep}}{2\Delta^3|N^F(v)|L'^2}\leq \frac{C^{rep}}{|N^F(v)|^2\Delta L'^2},
\end{align*}
where the second inequality is by $|N^F(v)|\leq 2\Delta^2$. Therefore $S$ satisfies functional inequality with $C=C^{rep}/2\Delta^3=\Omega(1/L)$.

To show the first inequality~\eqref{eq:1}, we observe that the sum on the left-hand side over the edges can be decomposed to summing over the generalized repetition code $D(w)$ with center $w$ in the seam $\calM$, as shown in the orange part of~\Cref{fig:seam}. We would apply the functional inequality of generalized repetition code to each $D(w)$.

 From the observation above, we can obtain the following inequality
\begin{align*}
    \sum_{\{x,y\}\in E}|g(x)-g(y)|&=\sum_{w\in \calM}\sum_{\{x,y\}\in E(D(w))}|g(x)-g(y)|\\
    &\geq \sum_{w\in \calM}\frac{C^{rep}}{\Delta L'}\sum_{x,y\in D(w)}|g(x)-g(y)|,
\end{align*}
where the inequality is by the functional inequality of the generalized repetition code.

To show the second inequality~\eqref{eq:2}
\begin{align*}
    \sum_{w\in\calM}\sum_{x,y\in D(w)}|g(x)-g(y)|&\geq \frac{1}{|N^F(v)|^2L'}\sum_{x,y}|g(x)-g(y)|,
\end{align*}
we can first prove that for two neighboring faces $f_1, f_2$,
\begin{align}
    \sum_{x\in f_1,y\in f_2}|g(x)-g(y)|&\leq \sum_{z\in f_2}L'\left(\sum_{x\in f_1\cap D(w_{z1})}|g(x)-g(z)|+\sum_{y\in f_2\cap D(w_{z2})}|f(z)-f(y)|\right)\nonumber\\
    &\leq L'\left(\sum_{w\in e_1}\sum_{x\in D(w)\cap f_1, z\in D(w)\cap f_2}|g(x)-g(z)|+\sum_{w\in e_2}\sum_{y,z\in D(w)\cap f_2}|g(y)-g(z)|\right),\label{eq:3}
\end{align}
where $z$ and $x$ are on the same horizontal line, and $z$ and $y$ are on the same vertical line. Readers can refer to~\Cref{fig:xyz} for the definition of the notations used in the inequality. The first line is based on the triangular inequality, and the second line is by reorganizing the sum order.

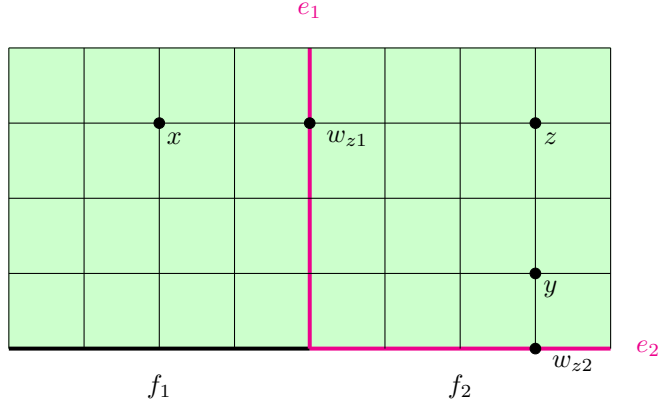
\begin{figure}

    \centering
    % \begin{subfigure}[b]{0.4\textwidth}
    %     \begin{tikzpicture}
    %         \foreach \row in {0,1,...,\rows}{
    %         \draw (0,\row)--(\rows-\row,\row);
    %         \draw (\row,0)--(\row,\rows-\row);
    %         \draw (0,\row)--(\row,0);
    %         \draw[dashed] (\row,\rows-\row)--(\row+0.5,\rows-\row+0.5);
    %         \draw (\row+0.5,\rows-\row+0.5) circle (2pt);
    %         }
    %     \end{tikzpicture}
    % \caption{One triangle in the GTSC;}
    % \label{fig:tri}
    % \end{subfigure}
        \begin{tikzpicture}
            \filldraw[green!20] (-4,0) rectangle (4,4);
            \draw[line width=1.5pt](-\rows,0)--(0,0);
            \draw[line width=1.5pt,color=magenta] (0,0)--(\rows,0);
            \draw[line width=1.5pt,color=magenta] (0,0)--(0,\rows);
            \foreach \row in {1,2,...,\rows}{
                \draw (-\rows,\row)--(\rows,\row);
                \draw (\row,0)--(\row,\rows);
                \draw (-\row,0)--(-\row,\rows);
            }
            \filldraw (-2,3) circle (2pt);
            \node at (-1.8,2.8) {$x$};
            \filldraw (3,3) circle (2pt);
            \node at (3.2,2.8) {$z$};
            \filldraw (3,1) circle (2pt);
            \node at (3.2,0.8) {$y$};
            \filldraw (0,3) circle (2pt);
            \node at (0.5,2.8) {$w_{z1}$};
            \filldraw (3,0) circle (2pt);
            \node at (3.5,-0.2) {$w_{z2}$};

            % \draw[dashed] (\row,\rows-\row)--(\row+1,\rows-\row);
            % \draw (\row+1,\rows-\row) circle (2pt);
            \node at (-2,-0.5) {{$f_1$}};
            \node at (2,-0.5) {{$f_2$}};
            \node at (0,4.5) {\textcolor{magenta}{$e_1$}};
            \node at (4.5,0) {\textcolor{magenta}{$e_2$}};
        \end{tikzpicture}

    \caption{The relation of $x,y,z$ we are considering when computing $C$ in the functional inequality. Please note that now each edge corresponds to a level 1 vertex, and each vertex corresponds to a level 0 vertex. The thickened edges are the seam of local structure $S$. We marked the $e_1$ and $e_2$ in our formula with red, and the faces $f_1$,$f_2$ with green.
    }
    \label{fig:xyz}
\end{figure}

 For each pair of faces, we can obtain the same bound. If we sum all possible face pairs together, we will sum over all possible $e_1, e_2$ in the seam $\calM$ on the right-hand side of the inequality~\eqref{eq:3}, giving us the following inequality:
\begin{align*}
    \sum_{f,f'\in N^F(v)}\sum_{x\in f,y\in f'}|g(x)-g(y)|\leq |N^F(v)|^2L'\sum_{w\in\calM}\sum_{x,y\in D(w)}|g(x)-g(y)|.
\end{align*}

Since the summation on the left hand side is same as the summation over $x, y \in V$, this implies inequality~\eqref{eq:2}.

For the bound on $C^{\partial}$, we can first assume that every face $f$ has the boundary vertices. Since every local complex $S$ has at least one face with boundary, by adding boundary vertices to dummy faces, the $|V^\partial|$ would increase by at most a constant factor, and the sum $\sum_{x\in V^\partial,y\in V}|g(x)-g(y)|$ will only increase. Thus if we prove a lower bound for the full boundary case, we can get general lower bounds on $C^{\partial}$ by losing a constant factor.

The proof has a similar flavor as the proof above. For the boundary vertices $x\in f_1^\partial$ and vertices $y\in f_1$, by triangular inequality, we have
\begin{align}\label{eq:4}
        \sum_{x\in f_1^\partial,y\in f_1}|g(x)-g(y)|\leq \sum_{z\in f_1}\left(L'|g(x)-g(z)|+\sum_{y\in f_1}|g(y)-g(z)|\right),
\end{align}
where $z$ and $x$ are on the same horizontal/vertical line, and $z$ and $y$ are on the same vertical/horizontal line.

\begin{figure}
    \centering
     \begin{tikzpicture}
            \filldraw [orange!20] (-4.1,4.2) rectangle (0.1,4.8);
            \filldraw [orange!20] (-4.8, -0.1) rectangle (-4.2,4.1);
            \filldraw[green!20] (-4,0) rectangle (4,4);
            \draw[line width=1.5pt,color=magenta](\rows,0)--(0,0);
            \draw[line width=1.5pt,color=magenta] (0,0)--(-\rows,0);
            \draw[line width=1.5pt,color=magenta] (0,0)--(0,\rows);
            \draw[line width=1.5pt,color=magenta,dashed] (0,4)--(0,4.5);
            \draw[line width=1.5pt,color=magenta,dashed] (-4,0)--(-4.5,0);
            \foreach \row in {1,2,...,\rows}{
                \draw (-\rows,\row)--(\rows,\row);
                \draw (\row,0)--(\row,\rows);
                \draw (-\row,0)--(-\row,\rows);
                \draw[dashed] (-\row,\rows)--(-\row,{\rows+0.5});
                \draw [dashed] (-\rows,\row)--({-\rows-0.5},\row);
            }
            \filldraw (-2,3) circle (2pt);
            \node at (-1.8,2.8) {$z$};
            \filldraw (3,3) circle (2pt);
            \node at (3.2,2.8) {$y$};
            \filldraw (-2,4.5) circle (2pt);
            \node at (-1.8,4.3) {$x$};
            \filldraw (0,3) circle (2pt);
            \node at (0.5,2.8) {$w_{z1}$};
            \filldraw (-2,0) circle (2pt);
            \node at (-1.5,-0.3) {$w_{z2}$};

            % \draw[dashed] (\row,\rows-\row)--(\row+1,\rows-\row);
            % \draw (\row+1,\rows-\row) circle (2pt);
            \node at (-2,-0.5) {{$f_1$}};
            \node at (2,-0.5) {{$f_2$}};
            \node at (0.3,4.5) {\textcolor{magenta}{$e_1$}};
            \node at (-4.2,-0.3) {\textcolor{magenta}{$e_2$}};
            \node at (4.2,-0.3) {\textcolor{magenta}{$e_3$}};
            \node at (-4.5,4.5) {\textcolor{orange}{$f_1^{\partial}$}};
        \end{tikzpicture}
    \caption{The relation of $x,y,z$ we are considering when computing $C^\partial$. Note that here $w_{z1}$ will correspond to the first term in the inequality~\eqref{eq:4}, and $w_{z2}$ will correspond to the second term. The case when $x\in f_1^\partial,y\in f_1$ is similar.}
    \label{fig:xyzbound}
\end{figure}
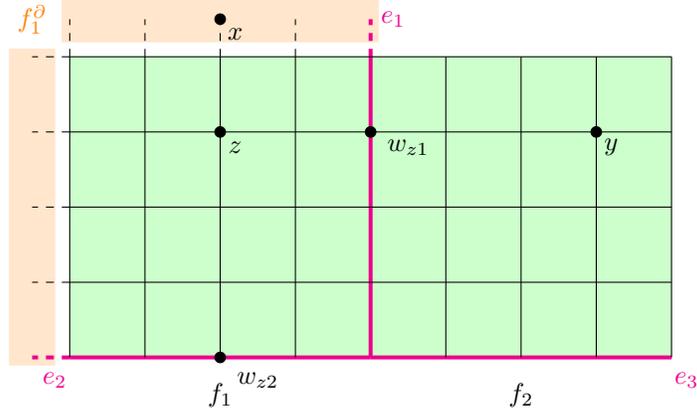

For $x\in f_1^\partial$ and $y\in f_2$ on neighboring faces $f_1$ and $f_2$, combining the inequality above, we have that
\begin{align*}
    \sum_{x\in f_1^\partial,y\in f_2}|g(x)-g(y)|&\leq\frac{1}{L'^2}\left(L'^2\sum_{x\in f_1^\partial,z\in f_1}|g(x)-g(z)|+2L'\sum_{z\in f_1,y\in f_2}|g(z)-g(y)|\right)\\
    &= \sum_{x\in f_1^\partial,z\in f_1}|g(x)-g(z)|+\frac{2}{L'}\sum_{z\in f_1,y\in f_2}|g(z)-g(y)|,
\end{align*}
where the first inequality is by the triangular inequality over all $z\in f_1$.
For the first summation term on the right-hand side, we apply inequality~\eqref{eq:4}, and for the second summation term, we apply inequality~\eqref{eq:3}. Combining the two inequalities, we have that
\begin{align*}
     \sum_{x\in f_1^\partial,y\in f_2}|g(x)-g(y)|&\leq L'\sum_{w\in E(f_1)}\sum_{\substack{x\in D(w)\cap f_1^\partial\\z\in D(w)\cap f_1}}|g(x)-g(z)|+2\sum_{w\in e_1}\sum_{z,y\in D(w)}|g(z)-g(y)|\\&+2\sum_{w\in e_3}\sum_{y,z\in D(w)\cap f_2}|g(y)-g(z)|,
\end{align*}
where we used $E(f_1)$ to denote the seam boundary of face $f_1$. For example, the branches $e_1$ and $e_2$ form the seam boundary of $f_1$ in~\Cref{fig:xyzbound}. Readers can refer to the figure for the relation between $x,y,z$.  The first inequality comes from triangular inequality, the second inequality is based on the previous inequality~\eqref{eq:3} on $f_1$ and $f_2$.

Summing all $(f_i^\partial,f_j)$ pairs, we have that
\begin{align*}
    \sum_{x\in V^\partial,y\in V}|g(x)-g(y)|&\leq |N^F(v)|^2L'\sum_{w\in\calM}\sum_{x\in D(w)^\partial y\in D(w)}|g(x)-g(y)|+3|N^F(v)|^2\sum_{w\in\calM}\sum_{x,y\in D(w)}|g(x)-g(y)|\\
    &\leq \frac{|N^F(v)|^2L'\Delta}{C^{rep,\partial}} \sum_{\{x,y\}\in E}|g(x)-g(y)|+\frac{3|N^F(v)|^2L'\Delta}{C^{rep}}\sum_{\{x,y\}\in E}|g(x)-g(y)|,
\end{align*}
 Since $|V^\partial|= 3|N^F(v)|L'\leq 6\Delta^2 L'$, we have that $C^\partial=\Omega(1/L)$.
\end{proof}

\begin{remark}
    Note that the result also holds for the subdivision of reasonable codes without dummy faces.
    The proof utilizes the fact that the link of a check is connected, which is automatically satisfied by the reasonable code.
\end{remark}

\subsection{Proof of~\Cref{thm:level1}}
In Section 5.4.2 of~\cite{linGeometricallyLocalQuantum2023}, the authors proved that the generalized surface code is also a $(\beta_1=\frac{2}{3L},\eta_1=\frac{L-1}{4L}$)-coboundary expander at level 1. The proof of level 1 expansion did not utilize the product structure of the generalized surface code in~\cite{linGeometricallyLocalQuantum2023}, we provide a sketch here, and refer interested readers to the original paper for detail.

For each $\hat{f}_1\in\F_2^{X(1)\cap S}$, we can decompose the support of $\hat{f}_1$ into different clusters $\hat{f}_{1,i}$ with disjoint support and violated checks $\delta_1\hat{f}_{1,i}$. The clusters have the following possible cases:
\begin{itemize}
    \item The cluster has no violated checks. The following cases will exclude this case.
    \item The cluster is not connected to the seam or the boundary.
    \item The cluster is connected to the boundary but not the seam.
    \item The cluster is connected to the seam.
\end{itemize}

For the first three cases, we can easily find some $f_{1,i}\in \hat{f}_{1,i}+B^1$ that satisfies the $(\beta_1,\eta_1)$ coboundary expansion. For the Case 4 when the cluster is connected to the seam, we would first push the path to the qubits next to the seam and set the final $f_{1,i}$ similarly as in Section 5.4.2. Note that the new complex has the same local structure as the generalized surface code around each seam, thus we can get a similar coboundary expansion result for $S$ in our construction following the similar reasoning.

\printbibliography
% End edit to here
%%%%%%%%%%%%%%%%%%%%%%%%%%%%%%%%%%%%%%%%%%%%%%%%%%%%%%%%%%%%%

\appendix

\section{Alternative Geometric Structures from Codes and Chain Complexes based on Nerve-like Construction}
\label{sec:app}

In \Cref{sec:2DStruc},
    we described a way to extract a 2D structure from any quantum LDPC code,
    which is the key observation in this paper.
In this section,
    we describe a similar geometrization process
    which applies to bounded degree chain complexes of arbitrary length and over arbitrary commutaive ring.

\subsection{Simplicial Complex from Chain Complex}

For the sake of being general,
    we introduce a more general notion of chain complex than the one discussed in \Cref{sec:chain-complex}.
These type of chain complexes can be induced naturally from algebraic topology.

\begin{definition}[Chain complex]
    Let $R$ be a commutative ring.
    A chain complex $X$ is a sequence of free $R$-modules $R^{X(i)}$ generated by sets $X(i)$, along with $R$-linear maps $\partial_i\colon R^{X(i)}\to R^{X(i-1)}$ known as boundary operators, where the boundary operators satisfy
    \begin{equation*}
        \partial_{i-1}\partial_i=0.\footnote{
        We choose the boundary operators over the coboundary operators only for cosmetic reasons, as this choice leads to a simpler sign assignments that we will encounter later.
        }
    \end{equation*}
\end{definition}
When $R = \F_2$, this reduces to the previous definition.
Other common choices of $R$ include $\F_p$ and $\ZZ$.
A less common choice of $R$ in the context of coding theory is $\ZZ / m \ZZ$ for some composite number $m$.
When $R$ is a field, a free $R$-module is exactly a vector space.

Nerve is the idea in category theory which treats
    ``objects'' as 0-simplices,
    ``arrows'' as 1-simplices,
    and ``composition of arrows'' as 2-simplices, etc.
For example, given three abstract objects $x, y, z$,
    with two arrows $f: x\to y$ and $g: y\to z$,
    they compose into another arrow $g f: x\to z$
    as shown in \Cref{fig:nerve}.
In the language of nerve,
    $x, y, z$ correspond to three 0-simplices,
    $x \xrightarrow[f]{} y, y \xrightarrow[g]{} z, x \xrightarrow[gf]{} z$ correspond to three 1-simplices,
    and the relation between $f, g, g f$,
        which will be denote as $x \xrightarrow[f]{} y \xrightarrow[g]{} z$,
        corresponds to a 2-simplex.

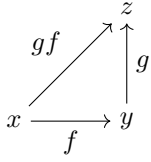
\begin{figure}
    \centering
    \begin{tikzpicture}[scale=1.5]
        \node (x) at (0,0) {$x$};
        \node (y) at (1,0) {$y$};
        \node (z) at (1,1) {$z$};
        \draw[->] (x) --node[auto,swap]{$f$} (y);
        \draw[->] (y) --node[auto,swap]{$g$} (z);
        \draw[->] (x) --node[auto]{$g f$} (z);
    \end{tikzpicture}
    \caption{The nerve construction from category theory.}
    \label{fig:nerve}
\end{figure}

We now apply this concept of nerve to a chain complex which will give the desired geometric structure.
The key is to view the chain complex as a collection of arrows.
The objects correspond to the basis of the free modules $\bigcup_i X(i)$.
The simple arrows correspond to the boundary maps $\partial_i: R^{X(i)} \to R^{X(i-1)}$.
Specifically, there is an arrow $x_i \to x_{i-1}$ for $x_i \in X(i), x_{i-1} \in X(i-1)$
    iff the corresponding entry $(x_i, x_{i-1})$ in $\partial_i$ is nonzero.
Finally, The full set of arrows consists of all possible compositions of the simple arrows.

Using the objects and the arrows described above,
    the nerve construction gives a $t$-dimensional simplicial complex,
    when the given chain complex has $t+1$ terms.
In particular, each $s$-dimensional simplex
    corresponds to a sequence of arrows $x_i \to x_{i-1} \to ... \to x_{i-s}$,
    where the entry $(x_j, x_{j-1})$ is nonzero for $\partial_j$.
We sometime include the values of the entries which will be denoted as
    $x_i \xrightarrow{r_i} x_{i-1} \xrightarrow{r_{i-1}} ... \xrightarrow{r_{i-s+1}} x_{i-s}$.
It is straightforward to see that if the chain complex has bounded degree,
    the resulting simplicial complex also has bounded degree,
    because there are only bounded number of ways to combine the arrows.

\subsection{Subdivide the Chain Complex Using the Structure of the Simplicial Complex}

We now demonstrate that this simplicial complex can be viewed as a geometrical realization of the chain complex
    by focusing on the subdivision process.
The subdivision process can be understood through two steps.
The first step is the ability to subdivide each $s$-simplex (with values) $x_i \xrightarrow{r_i} x_{i-1} \xrightarrow{r_{i-1}} ... \xrightarrow{r_{i-s+1}} x_{i-s}$.
The second step is to show that if we subdivide all simplices in the simplicial complex, the resulting subdivided simplices attach nicely and gives a new chain complex.
In a more abstract language,
    if we view individual simplices as local objects
    and the full simplex as a global object,
    the preceding description says that
    to subdivide a global object,
    all we need is the ability to subdivide a local object (the first step)
    and make sure the subdivided local components glue into a new global object (the second step).
The second step can be checked straightforwardly, so we will focus on the first step.

An intuitive way to think about $x_s \to x_{s-1} \to ... \to x_0$
    and its subdivision is to think of it as a $s$-simplex in a $s$-dimensional cell
    as shown in \Cref{fig:nerve-geometrical}.
In this context, $x_i$ corresponds to one of the $i$-cells,
    and the cells contain each other.
There is a sense that the $s$-simplex is like an incomplete cell within this larger $s$-dimensional cell,
    which causes it to have nontraditional boundaries.
In particular, the example in the figure is a $2$-simplex,
    which has three types of boundary, $2 \to 1$, $1 \to 0$, and $2 \to 0$.

\begin{figure}
    \centering
    \includegraphics[width=0.5\textwidth]{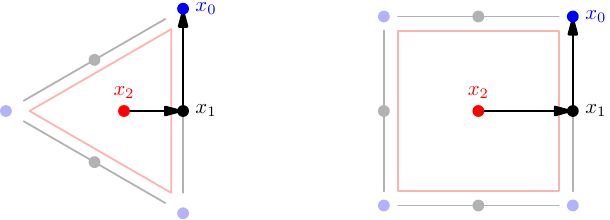}
    \caption{Viewing $x_2 \to x_1 \to x_0$ as an incomplete $2$-cell which is a part of a triangle or a square.}
    \label{fig:nerve-geometrical}
\end{figure}

The geometric subdivision should respect these boundary types
    as shown on the left of \Cref{fig:subdivision-general}.
We see that the bottom boundary is always the $2 \to 1$ type,
    the right boundary is always the $1 \to 0$ type,
    and the diagonal boundary is always the $2 \to 0$ type.
This geometric subdivision then naturally induces a family of arrows
    as shown on the right of \Cref{fig:subdivision-general}
    where the arrows are induced from the containment relation.
We denote the arrow structure of the subdivision
    as $Y_s \to Y_{s-1} \to ... \to Y_0$.

\begin{figure}
    \centering
    \includegraphics[width=0.97\textwidth]{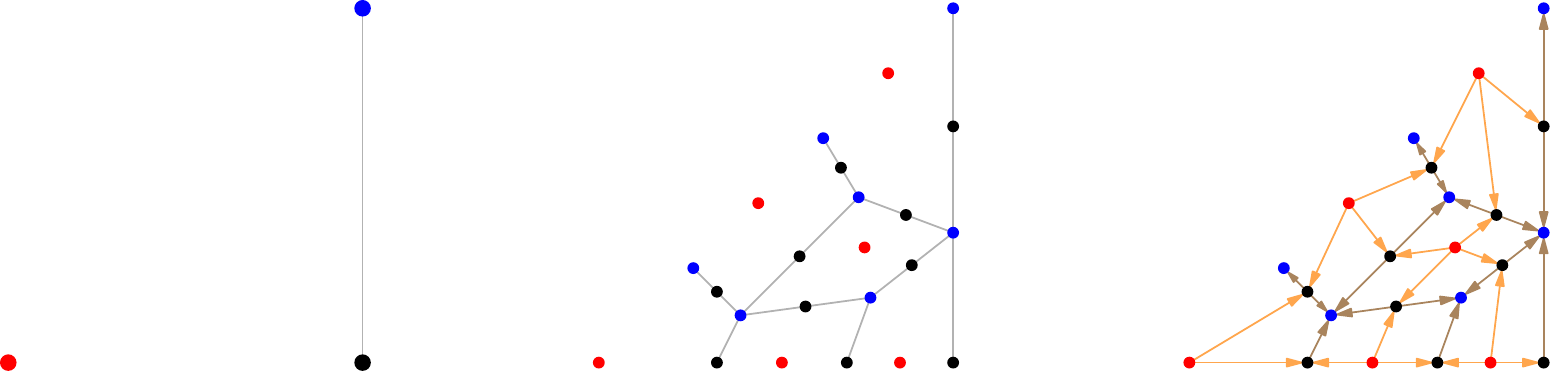}
    \caption{Left: The original incomplete $2$-cell.
             Middle: A geometric subdivision that respects the boundary types.
             Right: The corresponding arrow structure of the subdivision.
             The orange arrows are of $y_2 \to y_1$
             and the brown arrows are of $y_1 \to y_0$.}
    \label{fig:subdivision-general}
\end{figure}

To obtain a chain complex, we additionally need to assign values to the arrows of the subdivided structure.
When subdividing $x_s \xrightarrow{r_s} x_{s-1} \xrightarrow{r_{s-1}} ... \xrightarrow{r_{1}} x_0$,
    the arrow $y_i \to y_{i-1}$ with $y_i \in Y_i, y_{i-1} \in Y_{i-1}$
    will have value $r_i$ or $-r_i$.
The sign is determined from the standard sign convention from algebraic topology.
In particular, there is an orientation for each $i$-cell
    and the sign corresponds to whether the orientation agrees.
Note that orientations in the interior could be arbitrary,
    but the orientations on the boundaries have to be fixed
    so that the structures glue globally.
A consistent assignment for the orientations on the boundaries
    can be induced from
    the orientation of the original simplex $x_s \to x_{s-1} \to ... \to x_0$.
For each $i$, the orientation of the original incomplete $i$-cell $x_i$
    is set to be the direction $x_i \to x_{i-1} \to ... \to x_0$.
The cells along the boundary should follow similar orientations.
For example, say $s=2$ as in \Cref{fig:sign-assignment},
    the new $2$-cells on the boundary should follow the same orientation as the original $2$-cell.
    The new $1$-cells on the $2 \to 1$ boundary,
        should point away from the $2 \to 1$ boundary,
        similar to the original $1$-cell.
    The new $1$-cells on the $1 \to 0$ boundary,
        should point along the $1 \to 0$ boundary,
        similar to the original $1$-cell.
    The new $1$-cells on the $2 \to 0$ boundary,
        should point to the $2 \to 0$ boundary,
        similar to the original $1$-cell.

\begin{figure}
    \centering
    \includegraphics[width=0.97\textwidth]{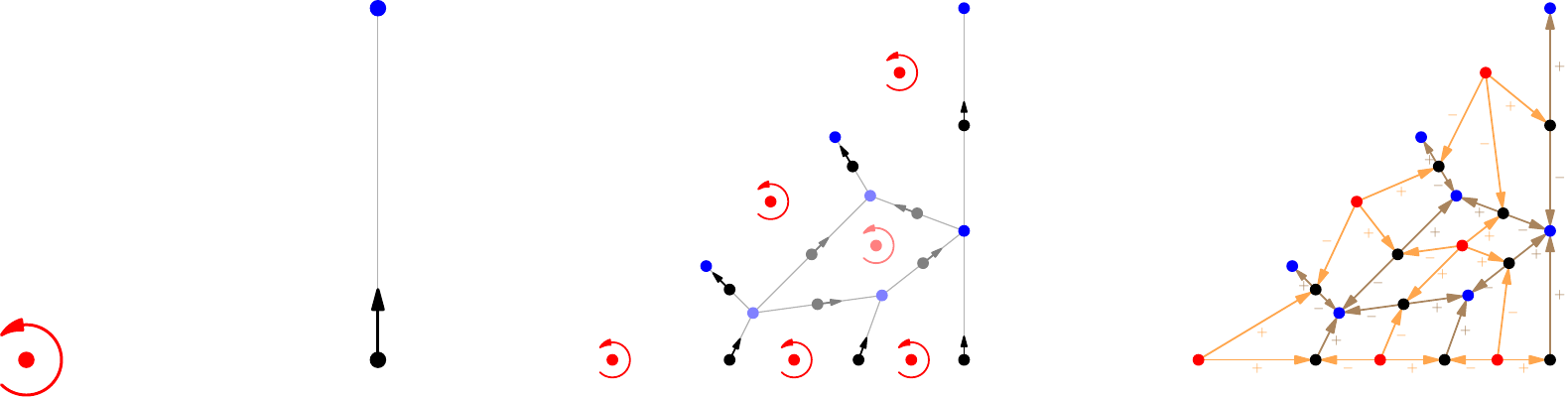}
    \caption{Left: The original incomplete $2$-cell with orientations.
             Middle: The geometric subdivision with orientations.
             Note that the orientations on the boundaries are fixed and are marked in opaque colors.
             Right: The corresponding sign assignments.}
    \label{fig:sign-assignment}
\end{figure}

\begin{remark}
    We note that an additional step is required for the new chain complex to be homotopic equivalent to the original chain complex.
    This step is the analog of the pairing process described in the main text.
    To see the necessity of such process,
        we consider the following $\F_2$ chain complex
    \begin{equation*}
        \F_2 \xrightarrow{\mathmakebox[3.5em]{\left[\begin{smallmatrix}
            1 & 1 & 1 & 1
        \end{smallmatrix}\right]}}
        \F_2^4 \xrightarrow{\mathmakebox[3.5em]{\left[\begin{smallmatrix}
            1 \\ 1 \\ 1 \\ 1
        \end{smallmatrix}\right]}}
        \F_2,
    \end{equation*}
    which corresponds to $4$ triangles that share an edge,
        as depicted in Left 1 of \Cref{fig:pairing-example}.

    The subdivision where each triangle is replaced by half of a surface code
        is depicted in Left 2 of \Cref{fig:pairing-example}.
    Unfortunately, this new chain complex is not homotopic equivalent to the original chain complex.
    One way to infer this is to notice that we introduced
        $5$ new elements in $X(0)$,
        $12$ new elements in $X(1)$,
        $5$ new elements in $X(2)$
        in the new chain complex.
    Since $5 - 12 + 5 \ne 0$, this process changes the Euler characteristic.
    Indeed, this process introduces new nontrivial homologies,
        for example, the vector with value $1$ on the circled elements in $X(1)$,
        as depicted in Right 2 of \Cref{fig:pairing-example}.

    This is where pairing come in.
    By pairing the two triangles on the left and pairing the two triangles on the right,
        we splits the structure where $4$ triangles share an edge
        into $2$ squares.
    This effectively splits the vertical cord into two,
        as depicted in Right 1 of \Cref{fig:pairing-example}.
    This leads to two additional elements one in $X(0)$ and one in $X(2)$.
    It is clear that the Euler characteristic is the same as the original chain complex as $6 - 12 + 6 = 0$.
    Indeed, the two chains complexes are now homotopically equivalent.

    \begin{figure}
        \centering
        \includegraphics[width=0.97\textwidth]{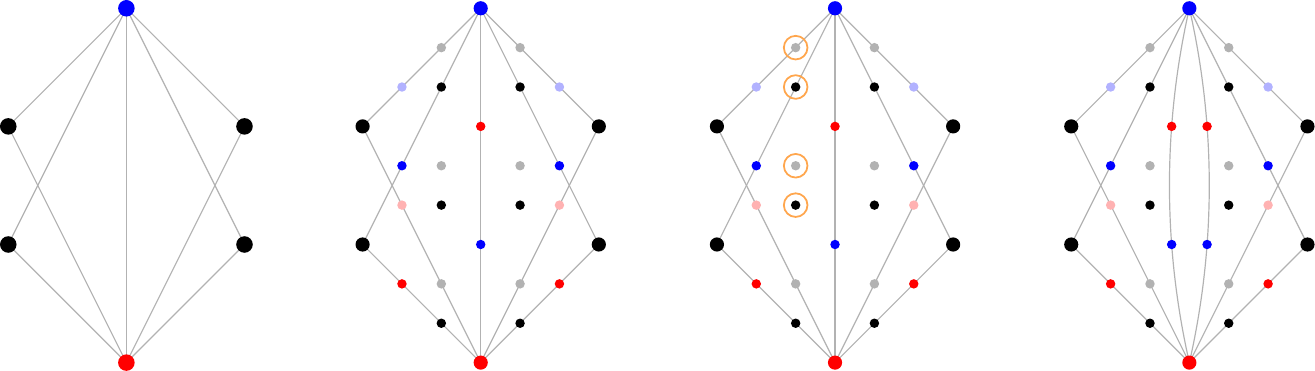}
        \caption{Left 1: The original chain complex.
                 Left 2: The naive subdivision without pairing.
                 Right 2: A nontrivial homology marked by the orange circles.
                 Right 1: The correct subdivision with pairing.}
        \label{fig:pairing-example}
    \end{figure}

    Since the pairing process is more complicated for longer chains,
        the details will be described elsewhere.
    After including pairing in the subdivision process,
        we believe that the new chain complex will be homotopic equivalent to the original chain complex,
        similar to what happens to the subdivision of a CW complex.
    If the claim is true, this would support the interpretation that the new chain complex can be viewed as a cellulation of the original chain complex.
\end{remark}

\end{spacing}
\end{document}